\documentclass[10pt]{amsart}
\usepackage[utf8]{inputenc}
\usepackage[english]{babel}
\usepackage{amssymb}
\usepackage{bm}
\usepackage{bbm}
\usepackage{xcolor}
\usepackage{graphicx}
\usepackage{esint}
\usepackage{tikz}
\usepackage{enumitem}
\usetikzlibrary{arrows}
\usetikzlibrary{snakes}

 \newtheorem{thm}{Theorem}
 \newtheorem{lem}{Lemma}[section]
  \newtheorem{cor}[lem]{Corollary}
 \newtheorem{prop}[lem]{Proposition}
 \newtheorem{defi}{Definition}
 \newtheorem{rem}{Remark}
 \numberwithin{equation}{section}

\usepackage{color}

\newcommand{\RR}{\mathbb{R}}
\newcommand{\NN}{\mathbb{N}}
\newcommand{\1}{\mathbbm{1}}
\newcommand{\TT}{\mathbb{T}}

\newcommand{\vb}{|}
\newcommand{\norm}{\|}
\newcommand{\curl}{\mathrm{curl}\,}
\newcommand{\di}{\mathrm{div}\,}
\renewcommand{\j}{E}
\newcommand{\oj}{\overline E}
\newcommand{\tj}{\tilde E}

\newcommand{\bj}{\bar E}
\newcommand{\vnu}{\bm{\nu}}

\def\({\left(}
\def\){\right)}
\def\1{\mathbf{1}}

\def\a{{\alpha}}
\def\ep{\varepsilon}
\def\curl{{\rm curl\,}}
\def\supp{\mathrm{Supp}\,}

\def\div{\mathrm{div} \ }
\def\dist{\mathrm{dist} \ }
\def\E{\Sigma}
\def\g{\gamma}
\def\hal{\frac{1}{2}}

\def\indic{\mathbf{1}}

\def\j{E}
\def\l{{\ell}}

\def\nab{\nabla}

\def\wb{{\overline{m}}}
\def\bw{{\underline{m}}}

\def\mz{\mathbb{Z}}
\def\mn{\mathbb{N}}
\def\mr{\mathbb{R}}

\def\om{\Omega}

\def\p{\partial}

\def\ro{\rho}
\def\r|{\right|}
\def\sm{\setminus}

\def\vp{\varphi}

\def\Xint#1{\mathchoice
   {\XXint\displaystyle\textstyle{#1}}%
   {\XXint\textstyle\scriptstyle{#1}}%
   {\XXint\scriptstyle\scriptscriptstyle{#1}}%
   {\XXint\scriptscriptstyle\scriptscriptstyle{#1}}%
   \!\int}
\def\XXint#1#2#3{{\setbox0=\hbox{$#1{#2#3}{\int}$}
     \vcenter{\hbox{$#2#3$}}\kern-.5\wd0}}
\def\fint{\Xint-}

\begin{document}

\title[Renormalized energy equidistribution in 2D Coulomb systems]{Renormalized energy equidistribution and local charge balance in 2D Coulomb systems}
\author{Simona Rota Nodari and Sylvia Serfaty}
\date{July 11, 2013}

\begin{abstract}
We consider two related problems: the first is the minimization of the ``Coulomb renormalized energy" of Sandier-Serfaty, which corresponds to the total Coulomb interaction of point charges in a uniform neutralizing background (or rather variants of it). The second corresponds to the minimization of the Hamiltonian of a two-dimensional ``Coulomb gas" 
or ``one-component plasma", a system of $n$ point charges with Coulomb pair interaction, in a confining potential (minimizers of this energy also correspond to ``weighted Fekete sets"). In both cases we investigate the microscopic structure of minimizers, i.e. at the scale corresponding to the interparticle distance. We show that in any large enough microscopic set, the value of the energy and the number of points are ``rigid" and completely determined by the macroscopic density of points.  In other words, points and energy are ``equidistributed" in space (modulo appropriate scalings).  The number of points in a ball is in particular known up to an error proportional to the radius of the ball. We also prove a result on the maximal and minimal distances between points.
Our approach involves fully exploiting the minimality by  reducing to  minimization problems with fixed boundary conditions posed on smaller  subsets.\end{abstract}

\maketitle

\section{Introduction}
The ``renormalized energy" $W$, or more precisely Coulomb renormalized energy, was introduced in  \cite{sandierserfaty} where it appeared as a limiting energy for vortices in the Ginzburg-Landau model of superconductivity. It provides a way of computing a total Coulomb interaction for an {\it infinite} number of point charges in the plane, ``screened" by a constant density charge of the opposite sign, system which is also  called a  ``one-component plasma" or a ``jellium" in physics. The renormalized energy is computed as a ``thermodynamic limit", i.e. by taking averages of a certain energy computed on boxes whose sizes tend to $+\infty$, for a complete definition see below.
In \cite{sandierserfaty} it was conjectured that the ``Abrikosov" triangular lattice $\mathbb{Z}+ \mathbb{Z}e^{i\frac{\pi}{3}}$, properly scaled, achieves the minimum of $W$. This is supported by the fact that the triangular lattice  is proven to achieve the minimum of $W$ among ``simple" lattices of same volume, and in agreement with the observations in superconductors or triangular lattices of vortices, named in that context ``Abrikosov lattices". Proving rigorously that the triangular lattice achieves the minimum among all possible configurations remains a completely open question.

In \cite{sandierserfatygas} the same renormalized energy $W$ was also extracted as a limit for the minimization of the Hamiltonian associated to a two-dimensional Coulomb gas (also called ``one-component plasma"), 
\begin{equation}\label{wn}
w_n(x_1,\dots, x_n)= -\sum_{i\neq j} \log |x_i-x_j| + n\sum_{i=1}^n V(x_i)\end{equation}
where $x_1, \dots, x_n$ are $n$ points in the plane and $V$ is a confining potential growing sufficiently fast at infinity, and $n\to \infty$.  Due to its link with random matrices and determinantal processes, this Hamiltonian is particularly of interest when considered in a statistical mechanics setting, i.e.  with temperature, and its minimizers or ground states essentially correspond to the states in the limit of  zero temperature.

Since the question of identifying the minimum and minimizers of $W$ seems too hard, one can try to prove a weaker statement, namely that for a minimizer, the energy is ``equidistributed" i.e. asymptotically the same on any large enough square, regardless of where it is centered. This is inspired by such results in a paper of Alberti-Choksi-Otto \cite{aco} obtained for a somewhat similar energy arising in the context of the ``Ohta-Kawasaki model", for which it is also expected that the minimizers are periodic.

We will prove such a result here. However, the function $W$  defined in \cite{sandierserfaty} being defined as an average over squares centered at the origin and of size tending to infinity, it does not feel a compact perturbation of the configuration.  In some sense, $W$ is too ambiguous. Hence one cannot hope to prove the desired statement, unless  some boundary conditions  are fixed. One of the goals of this paper is also to study alternate minimization problems (with boundary condition, or with periodic boundary conditions), compare them, and show the result for these problems.

Another question we address, and which is closely related, is the equidistribution of the points in minimizers of $W$. Again, under appropriate boundary conditions, we will show that large boxes contain asymptotically the same number of points, simply proportional to their area and with an error proportional to the width of the box only, regardless of where they are centered. These results are optimal and point towards agreement with the idea of minimizers having some periodic behavior.

Our approach, designed to prove equidistribution of points and energy for $W$, works equally well to prove the same for minimizers of the Coulomb gas energy  \eqref{wn} in the limit $n\to \infty$. Note that such points are ``weighted Fekete sets" (for reference see \cite{safftotik}) and have attracted attention as such. It is known that these points concentrate on a subset of the plane, that we will denote $\E$, with a density proportional to  $\mu_0$, corresponding to the ``equilibrium measure", both being functions of the potential $V$ (for more details see below). We assume here that $\E$ is compact and has a regular enough boundary.   More precisely,  it is known that the leading order behavior for minimizers of $w_n$ is given by
\begin{equation}\label{eqcvminim}
\frac{1}{n}\sum_{i=1}^n \delta_{x_i} \rightharpoonup \mu_0\end{equation} in the weak sense of probability measures. The macroscopic behavior of the points is thus understood, but the distances  between neighboring points is typically $n^{-1/2}$, this is what we call the {\it microscopic scale}.   In \cite{sandierserfatygas} it was  proven that  the distribution of the points at the microscopic scale, or equivalently after blow-up by $\sqrt{n}$ and after letting $n \to \infty$, is governed by the renormalized energy $W$, but this was done via a ``probabilistic" or averaged formulation, which we can completely bypass here.

In  \cite{aoc}, Ameur and Ortega-Cerd\`a proved that such points satisfy an asymptotic equidistribution property at the microscopic scale:  a given microscopic  ball $B$  of radius $R/\sqrt{n}$ contains $\sim n\int_{B}\mu_0(x)dx$ as $n\to \infty$ then $R \to \infty$.
More precisely they showed, through a method based on ``Beurling-Landau densities" and ``correlation kernel estimates", that
\begin{equation}\label{aoc}
\limsup_{R\to \infty} \limsup_{n\to \infty} \frac{D(x_n, R)}{R^2} = 0\end{equation}
as long as $\dist (x_n,  \partial (\sqrt{n} \E)) \ge \log^2 n$, where
 \begin{equation}\label{scalingr2}
 D(x,R):=\# \( \{x_i\} \cap B(x/\sqrt{n}, R/\sqrt{n})\) - n\mu_0(B(x/\sqrt{n}, R/\sqrt{n})).\end{equation} Here, we prove a result which is a bit stronger  because it says
 that 
 \begin{equation}\label{srequi}
 \limsup_{R\to \infty} \limsup_{n\to \infty} \frac{D(x_n, R)}{R} <+\infty\end{equation} with balls replaced by squares in the definition of $D$. This
  obviously implies \eqref{aoc}, and seems to be optimal since it cannot be improved even if the points form a periodic lattice. The result is however, for now, limited to points $x_n$ such that $\dist (x_n,   \partial (\sqrt{n} \E)) \ge n^\beta$ for some power $\beta\in (0,1)$ (we believe we could get any $\beta>0$ but we did not pursue it here). In \cite{lev} an estimate similar to \eqref{srequi} is proven for Fekete points on complex manifolds.
  In contrast with \cite{aoc,lev} our method 
  is simply energy-comparison based. 
  

 In addition, we also prove that the ``renormalized energy" $W$ is  equidistributed at the microscopic scale in arbitrary square, again staying sufficiently far away from $\p \E$. This hints again towards the Abrikosov lattice, since it is expected that minimizers of $W$ look like it.  It also improves on the result of \cite[Theorem 2]{sandierserfatygas},  where it was established that minimizers of $w_n$, seen after blow-up at scale $\sqrt{n}$  around a point in $\E$, tend to minimize $W$  for {\it almost every} blow-up center. In other words, it was an averaged result. Here what we show proves that for true minimizers, this holds after blow-up around {\it any} point sufficiently inside $\E$. At the same time it dispenses  with having to use the ``probabilistic"  framework of \cite{sandierserfaty} based on the ergodic theorem. 

Again, our method  is energy and comparison-based, as in \cite{aco}, and relies on the tools from \cite{sandierserfaty}.

Let us finally mention that minimizers of the Coulomb gas energy (i.e. \eqref{wn} but with $-\log $ replaced by the appropriate Coulomb kernel) in dimensions different from $2$ is also of interest. In dimension $1$, they are essentially completely understood, cf. \cite{kunz,bl,am}. In higher dimension, we refer to \cite{rs} for recent results.  

We next state the definitions and our  results  more precisely.

\subsection{The Renormalized Energy: Definitions}
 In this subsection, we give the exact definition of this renormalized energy. We follow here the conventions of \cite{sandierserfatygas}, in particular,  compared to \cite{sandierserfaty} the vector-fields have been rotated by $\pi/2$. We also give more general definitions, relative to bounded domains, that we will need here.

For a given set $\Omega $ in the plane, $|\Omega|$ will denote its Lebesgue measure.
In what follows, $K_R$ denotes the square $[-R,R]^2$, and $K_R(x)= x +K_R$ the square  centered at $x$ and of sidelength $2R$.

\begin{defi}\label{defAmOmega}
Let $m$ be a positive number and $\Omega$ a  simply connected open  subset of $\RR^2$. Let $\j$ be a vector field in $\Omega$. We say $\j$ belongs to the admissible class $\mathcal A_m(\Omega)$ if
\begin{equation}\label{eqdefAmOmega}
\di \j=2\pi(\nu-m),\quad \curl \j=0\quad \mbox{in}\ \Omega\\
\end{equation}
where $\nu$ has the form
$$
\nu=\sum_{p\in\Lambda}\delta_p\quad\mbox{for some discrete set}\ \Lambda\subset {\Omega}
$$
and if $\Omega=\RR^2$ we require
\begin{equation}\label{eqcondnu}
\frac{\nu(K_R)}{\vb K_R\vb}\quad\mbox{is bounded by a constant independent of}\ R>1.
\end{equation}
\end{defi}
In this paper we will not make use of the condition \eqref{eqcondnu} borrowed from the definition of \cite{sandierserfaty}.

Here the physical interpretation of  $\j$ is that it is like an  ``electric field" generated by the point charges at $\Lambda$ which are screened by an opposite uniform background of density $m$.

This can be generalized to charges that are  screened by a non-uniform background.
\begin{defi}\label{defBrhoOmega}
Let  $\Omega$ be a simply connected open subset of $\RR^2$, and $\rho$  be a nonnegative $L^\infty(\Omega)$ function. Let $\j$ be a vector field in $\Omega$. We say $\j$ belongs to the admissible class $\mathcal B_\rho(\Omega)$ if
\begin{equation}\label{eqdefBrhoOmega}
\di \j=2\pi(\nu-\rho),\quad \curl \j=0\quad \mbox{in}\ \Omega\\
\end{equation}
where $\nu$ has the form
$$
\nu=\sum_{p\in\Lambda}\delta_p\quad\mbox{for some discrete set}\ \Lambda\subset {\Omega}.
$$
\end{defi}

\begin{rem} If $\rho$ is constant then $\mathcal A_\rho(\Omega)=\mathcal B_{\rho}(\Omega)$.
Moreover, we denote $\mathcal A_m=\mathcal A_m(\RR^2)$ and $\mathcal B_\rho=\mathcal B_\rho(\RR^2)$.
\end{rem}

\begin{defi}Let $\Omega$ be a  simply connected open  subset of $\RR^2$, and $\rho$ 
 a nonnegative $L^\infty(\Omega)$ function.  Let $\j$ be a vector field in $\Omega$ such that
\begin{equation}\label{eqdefBrhoOmegabis}
\di \j=2\pi(\nu-\rho)\quad \mbox{in}\ \Omega\\
\end{equation}
where $\nu$ has the form
$$
\nu=\sum_{p\in\Lambda}\delta_p\quad\mbox{for some discrete set}\ \Lambda\subset \Omega.
$$
Moreover, assume that $\curl \j=0$ in $\cup_{p\in\Lambda}B(p,\eta_0)$ for some small $\eta_0>0$.
Then for any continuous  nonnegative function $\chi$, we let
\begin{equation}\label{defWchi}
W(\j,\chi)=\lim_{\eta\to 0}\left(\frac{1}{2}\int_{\Omega\backslash \cup_{p\in\Lambda}B(p,\eta)}\chi\vb \j\vb^2+\pi\log\eta\sum_{p\in\Lambda}\chi(p)\right).
\end{equation}
\end{defi}

For any family of sets $\{\mathbf U_R\}_{R>0}$ in $\RR^2$ we use the notation $\chi_{\mathbf U_R}$ for positive cutoff functions satisfying for some constant $C$ independent of $R$,
\begin{equation}
	\label{defcutoff}
\chi_{\mathbf U_R} \le 1, \quad	\vb \nabla \chi_{\mathbf U_R}\vb \le C,\quad \mathrm{supp}(\chi_{\mathbf U_R})\subset \mathbf U_R,\quad \chi_{\mathbf U_R}(x)=1\ \text{if}\ d(x,\mathbf{U}_R^c)\ge 1.
\end{equation}
We will use this mainly for squares. In the rest of the paper, when we write $\chi_{K_L}$ we will always mean any nonnegative function satisfying \eqref{defcutoff} relative to the square $K_L$.

\begin{defi}\label{defrenenergy} The renormalized energy $W$ is defined, for $\j\in \mathcal A_m$, by
\begin{equation}
	\label{defW}
	W(\j)=\limsup_{R\to+\infty}\frac{W(\j,\chi_{K_R})}{\vb K_R \vb},
\end{equation}
with $\{\chi_{K_R}\}_R$ satisfying \eqref{defcutoff} for the family of squares $\{K_R\}_{R>0}$.
\end{defi}
Thus $W$ is defined as an energy per unit volume, where the energy $\int |E|^2 $ needs to be computed in a ``renormalized way" \`a la Bethuel-Brezis-H\'elein \cite{bbh}, according to \eqref{defWchi}, in order to remove the  divergent contribution  of the singularity in $1/r$  that $E$ carries around  each point $p \in \Lambda$.

More precisely, if $\j\in \mathcal A_m(\Omega) $ or $\mathcal B_\ro (\om)$ then we may write $\j= - \nabla H$ for some $H$ satisfying $$- \Delta H= 2\pi\Big(\sum_{p\in \Lambda} \delta_p-\ro\Big),$$ and we can check that $\j \in \mathcal C^1(\overline{\om} \sm \Lambda)$. Also (and even if $\curl \j=0$ only near the $p$'s) in  the neighborhood of each $p\in \Lambda$ we have the decomposition $\j(x)=\nabla \log \vb x-p\vb+f(x)$ where $f$ is $\mathcal C^1$, and it easily follows that the limit \eqref{defWchi} exists if $\chi$ is compactly supported. It also follows that $\j$ belongs to $L^p_{loc}(\om)$ for any $p<2$, and that taking the trace of $\j$ on any regular enough one-dimensional curve makes sense.

The following additional facts and remarks about $W$ are mostly borrowed from \cite{sandierserfaty}:
\begin{itemize}
	\item In the definition \eqref{defW}, the squares  $\{K_R\}_{R}$ can be replaced by other families of (reasonable) shapes $\{\mathbf U_R\}_{R}$, this yields a definition of a renormalized energy $W_U$, where the letter $U$ stands for the family $\{\mathbf U_R\}_R$. The minimizers and the value of the minimum of $W_U$ are independent of $U$ however.
	\item The value of $W$ does not depend on $\{\chi_{K_R}\}_R$  as long as it satisfies \eqref{defcutoff}.  The need for the cutoffs $\chi_{K_R}$ is to avoid problems with points that are on or very near the boundary, which would cause
$W(\j, \indic_{K_R})$ to be infinite.  In fact $W(\j, \indic_\om)$ makes sense (by natural extension of \eqref{defWchi}) if and  only if some boundary value is known   for $\j$  which excludes points on the boundary of $\Omega$.
When no such boundary value is known, we need to resort instead to $W(\j, \chi_{\Omega})$ where $\chi_\Omega$ is as in \eqref{defcutoff}.
	\item It is easy to check that if $\j$ belongs to $\mathcal A_m$, $m>0$, then $\j'=\frac{1}{\sqrt m}\j\left(\frac{\cdot}{\sqrt m}\right)$ belongs to $\mathcal A_1$ and
	\begin{equation}
		\label{eqscalingW}
		W(\j)=m\left(W(\j')-\frac{\pi}{2}\log m\right).
	\end{equation}
\item $W$ is bounded below and admits a minimizer over $\mathcal A_1$, hence also over  $\mathcal A_m $ by \eqref{eqscalingW}.
	 In particular,
	\begin{equation}
		\label{eqscalingminW}
		\min_{\mathcal A_m}W=m\left(\min_{\mathcal A_1}W-\frac{\pi}{2}\log m\right).
	\end{equation}
	In what follows, we  sometimes  denote for shortness
	\begin{equation}
		\label{defsigma*m}
		\sigma^{*}_m=\min_{\mathcal A_m}W.
	\end{equation}
	\item Because the number of points and the volume  are infinite when $\om = \mr^2$, the interaction over large balls needs to be normalized by the volume, as in a thermodynamic limit. Thus $W$ does not feel compact perturbations of the configuration of points. Even though the interactions are long-range, this is not difficult to justify rigorously.
	\item When the set of points $\Lambda$ is itself exactly a lattice $\mathbb Z\vec{u}+\mathbb Z\vec{v}$ then $W$ can be expressed explicitly through the Epstein Zeta function of the lattice, cf. \cite{sandierserfaty}. Moreover, using results from number theory, one finds  (cf.  \cite[Theorem 2]{sandierserfaty}) that the unique minimizer of $W$ over lattice configurations of fixed volume is the triangular lattice. This supports the conjecture that the Abrikosov triangular lattice is a global minimizer of $W$, with a slight abuse of language since $W$ is not a function of the points, but of their associated electric  field $\j$.
\item $W$ can be relaxed as a function of the points only by setting
$$\mathbb W(\nu)= \inf_{\j\in \mathcal A_m, \eqref{eqdefAmOmega}\ \text{holds}} W(\j).$$
This defines a measurable function of $\nu$, see \cite{sandierserfatygas}.
\end{itemize}
Here, in order to describe the local behavior of minimizers of $W$, we need to impose some conditions: we will consider either Dirichlet type boundary conditions, or periodic boundary conditions.  Otherwise, since  $W$ does not  feel compact perturbations of the points, it is impossible to locally characterize  minimizers of $W$ itself.
We will prove however that the questions of minimizing $W$ over $\mathcal A_m$ or minimizing $W$ over  domains with boundary or periodic conditions, become equivalent as the size of the domains tend to $+\infty$. This is part of the content of our results below.

We next define what we mean by boundary conditions.
For any $L>0$, we recall that 
$K_L(a)=a+[-L,L]^2$,
with the understanding that if the argument $a$ is absent, we take $a=0$.
In all that follows $\vnu$ denotes the outer unit normal to a set.
Let $p\in (1,2)$, $a\in\mathbb R^2$, $L>0$ and  $\varphi\in L^p(\partial K_L(a))$. Let $M>0$ and $\gamma\in\left(\frac{3-p}{2},1\right)$.
Let $\rho$ be a nonnegative $L^\infty(K_L(a))$ function.
We introduce the following hypotheses:
\begin{equation}
\label{h1boundary}
\frac{1}{2\pi}\int_{\partial K_L(a)}\varphi+\int_{K_L(a)}\rho(x)\,dx \in \NN;
\tag{$\mathrm{HB}_1$}
\end{equation}
\begin{equation}
\label{h2boundary}
\int_{\partial K_L(a)}\vb\varphi\vb^p\le ML^{2-\gamma}.
\tag{$\mathrm{HB}_2$}
\end{equation}
We also consider the sets of vector fields with normal trace $  \vp$ on $\p K_{L(a)}$:\begin{equation}\label{eqdefAmphi}
\mathcal A_{m,\varphi}(K_L(a))=\left\{\j\in \mathcal A_m(K_L(a)) \vb\, \j\cdot \vnu=\varphi \mbox{ on } \partial K_L(a)\right\},
\end{equation}
\begin{equation}\label{eqdefBphirho}
\mathcal B_{\rho,\varphi}(K_L(a))=\left\{\j\in \mathcal B_\rho(K_L(a)) \vb\, \j\cdot \vnu=\varphi \mbox{ on } \partial K_L(a)\right\}.
\end{equation}
Note that for $\mathcal A_{m, \vp} (K_L(a))$ and $\mathcal B_{\ro, \vp} (K_L(a))$ to be nonempty, we need the compatibility condition \eqref{h1boundary} to be satisfied.
\begin{rem}\label{reminterseclambda}
We have $\j\in \mathcal B_{\ro,\varphi}(K_L)$ with  $\varphi\in L^p(\partial K_L)$ for some $p\in (1,2)$ if and only if  $\Lambda\cap \partial K_L=\emptyset$. Indeed, in a neighborhood of $q\in\Lambda$, we have the decomposition $\j(x)=\nabla \log\vb x-q\vb+f(x)$ where $f$ is continuous. But one may check that if $q\in \p K_L$, then  $\nab\log \vb x-q\vb \cdot \vnu$ is not an $L^p(\p K_L)$ function (it is a distribution). 
 This proves  $\Lambda \cap \partial K_L= \emptyset$. Conversely, if $\Lambda\cap \p K_L= \emptyset$,  we  deduce that $\j$ is $C^1$ near $\p K_L$ hence the trace $\j \cdot \vnu $ makes sense in $L^p(\p K_L)$.
So we could as well add the requirement $\Lambda \cap \p K_L(a)=\emptyset$ in the definitions of $\mathcal A_{m,\vp}$ and $\mathcal B_{\ro, \vp}$.
In addition this justifies that $W(\j, \1_{K_L(a)})$ makes sense for such vector fields.

\end{rem}

\begin{defi} Let $p\in (1,2)$, $a\in\mathbb R^2$, $L>0$ and  $\varphi\in L^p(\partial K_L(a))$. Let $\rho$ be a nonnegative $L^\infty(K_L(a))$ function. Assume that  \eqref{h1boundary} is satisfied.
We define
\begin{equation}\label{eqdefsigmaphi}
\sigma_\vp (K_L(a); \ro) \ \text{resp.} \ \sigma_\varphi(K_L(a);m)=\min\limits_{\j\in \mathcal B_{\rho, \vp} (K_L(a)) \ \text{resp.} \  \mathcal A_{m,\varphi}(K_L(a))}\frac{W(\j,\1_{K_L(a)})}{\vb K_L(a)\vb},
\end{equation}
according to whether $\rho$ is equal to a constant $m$ or not.
\end{defi}

\begin{rem}\label{remexistencemin}		
	We will prove later that if $\varphi$ satisfies \eqref{h2boundary}, then the minimum of $W(\j,\1_{K_L(a)})$ over $\mathcal B_{\rho,\varphi}(K_L(a))$ is achieved for all $a\in \RR^2$ and $L>0$ fixed, i.e.  $\sigma_\varphi(K_L(a); m)$ and $\sigma_\varphi(K_L(a); \rho)$ are well defined.
\end{rem}		



Finally, we turn to the periodic setting.
When the set of points $\Lambda$ is periodic with respect to some lattice $\mathbb Z\vec{u}+\mathbb Z\vec{v}$ it can be viewed as a set of $n$ points $a_1,\ldots,a_n$ over the torus $\mathbb T_{(\vec{u},\vec v)}=\mathbb R^2/(\mathbb Z\vec{u}+\mathbb Z\vec{v})$, and we can give an explicit formula for $W$ in terms of $a_1, \dots, a_n$ (see \cite{sandierserfaty,sandierserfatygas}).

\begin{defi}\label{defAmper}
Let $m$ be a positive number, and  $(\vec{u},\vec v)$ a basis of $\mathbb R^2$. Let $\j$ be a vector field in $\mathbb T=\RR^2/(\mz  \vec{u} + \mz \vec{v}) $. We say $\j$ belongs to the admissible class $\mathcal A_{m,per}(\mathbb T)$ if
\begin{equation}\label{eqdefAmper}
\di \j=2\pi(\nu-m),\quad \curl \j=0\quad \mbox{in}\ \mathbb T\\
\end{equation}
where $\nu$ has the form
$$
\nu=\sum_{p\in\Lambda}\delta_p\quad\mbox{for some discrete set}\ \Lambda\subset \mathbb T.
$$
\end{defi}

It is clear that such vector fields $\j$ exist if and only if $\#\Lambda= m |\mathbb T| $.
 As shown in \cite{sandierserfaty}, if $\j\in \mathcal A_{m, per}(\mathbb{T})$ then $W(\j)$  (viewed as the renormalized energy of the periodic vector-field in all  of $\mr^2$) is given by
\begin{equation}\label{eqminper}
W(\j)=\frac{1}{\vb \mathbb T\vb}\lim_{\eta\to 0}\left(\frac{1}{2}\int_{\mathbb T \backslash \cup_{p\in \Lambda}B(p,\eta)}\vb \j\vb^2+\pi\#\Lambda\log\eta\right)=\frac{W(\j,\1_{\mathbb T})}{\vb \mathbb T\vb}
\end{equation}
For simplicity, we will restrict ourselves to square tori (but our approach works in general tori) and in the sequel, we denote $\mathbb T_{L}=\RR^2/(2L\mathbb Z)^2$.

\begin{defi} Let $m$ be a positive number, $L>0$ and $\mathbb T_{L}=\RR^2/(2L\mathbb Z)^2$ such that $m\vb \mathbb T_L\vb\in \mathbb N$.
We define
\begin{equation}\label{eqdefsigmaper}
\sigma_{per}(L;m)=\min\limits_{\j\in \mathcal A_{m,per}(\mathbb T_L)}\frac{W(\j,\1_{\mathbb T_L})}{\vb \mathbb T_{L}\vb}=\min\limits_{\j\in \mathcal A_{m,per}(\mathbb T_L)}W(\j).
\end{equation}
\end{defi}

\subsection{Main results on the renormalized energy}

\begin{thm}\label{thmmaincb}
 Let $p\in(1,2)$ and $m$ be a positive number. Then:
\begin{enumerate}
\item for all sequences of real numbers $L$ such that $m\vb K_L\vb \in \NN$, we have
\begin{equation}\label{eqlimsigma0cb}
\lim\limits_{L\to+\infty} \sigma_{0}(K_L;m)=\sigma_m^*;
\end{equation}
\item  given $\gamma\in\left(\frac{3-p}{2},1\right)$ and $M>0$, we have
\begin{equation}\label{eqlimsigmaphicb}
\lim\limits_{L\to+\infty} \sigma_{\varphi}(K_L;m)=\sigma_m^*.
\end{equation}
uniformly w.r.t. $\varphi$ such that \eqref{h1boundary} and \eqref{h2boundary} are satisfied in $K_L$.
\end{enumerate}

Moreover, there exists $\beta \in(0,1)$ such that the following holds: let $\j_\varphi$ be a minimizer for $\sigma_{\varphi}(K_L;m)$ and $\nu$ be associated via \eqref{eqdefAmOmega}.
Then there exists $\bar c>0$, $C>0$,   depending on $p$, $\gamma$, $m$, $M$,  such that for every $\l\ge \bar c$ and $a\in K_L$ such that if $d(K_\l(a),\partial K_L)\ge  L^\beta$, we have \begin{equation}\label{eqenergycb}
\left\vb\frac{W(\j_{\varphi},\chi_{K_{\l}(a)} )}{\vb K_{\l(a)}\vb}-\sigma_m^*\right\vb\le o(1)_{\l\to+\infty}
\end{equation}
 uniformly w.r.t. $\varphi$ such that \eqref{h1boundary} and \eqref{h2boundary} are satisfied in $K_L$, and
\begin{equation}\label{eqnumpointscb}
\left\vb   \nu(K_\ell(a))-m\vb K_\l(a)\vb \right\vb\le C \l.
\end{equation}

If in addition $\l$ and $\varphi$ are such that $\int_{\partial K_L\cap K_\l(b)}\vb \varphi\vb ^p\le M \l^{2-\gamma}$ for every $b\in \partial K_L$ then the result \eqref{eqenergycb}  holds for every $K_\l(a)\subset K_L$ (i.e holds up to the boundary).  If the assumption $\int_{\partial K_L\cap K_1(b)}\vb \varphi\vb ^p\le M $ holds   for some $M>0$  independent of $L$  and  every $b\in \partial K_L$, then both \eqref{eqenergycb} and \eqref{eqnumpointscb} hold  for every $\l \ge \bar{c}$ and every $K_\l(a)\subset K_L$.
\end{thm}

\begin{thm}\label{thmmainper}
Let $m$ be a positive number and $\mathbb T_{L}=\RR^2/(2L\mathbb Z)^2$. Then for all sequences of real numbers $L$ such that $m\vb \mathbb T_L\vb \in \NN$, we have
\begin{equation}\label{eqlimsigmaper}
\lim\limits_{L\to+\infty} \sigma_{per}(L;m)=\sigma_m^*.
\end{equation}
Moreover, let
 $\j$ be a minimizer for $\sigma_{per}(L;m)$ and $\nu$ be associated via \eqref{eqdefAmOmega}.
Then, there exists $\bar c>0$, $C>0$, depending only on $m$, such that for every $\l\ge \bar c$ and $a\in \TT_L$, we have
\begin{equation}\label{eqenergyper}
\left\vb\frac{W(\j,\chi_{K_{\l}(a)} )}{\vb K_{\l(a)}\vb}-\sigma_m^*\right\vb\le o(1)_{\l\to+\infty},
\end{equation}
\begin{equation}\label{eqnumpointsper}
\left\vb\nu(K_\l(a)) -m\vb K_\l(a)\vb \right\vb    \le C\l .
\end{equation}
\end{thm}
As announced, these results show that the minima of $W$ (as in Definition \ref{defrenenergy}), and of $W$ with periodic or fixed boundary conditions are asymptotically the same when the size of the boxes becomes large, and that minimizers have equidistributed energy and points, except possibly in a layer near the boundary in the case of Theorem~\ref{thmmaincb} --- this is however unavoidable, as it takes a certain distance for a configuration to ``absorb" or ``screen" the effect of a large or highly oscillating boundary data. These are analogues of the results proven in \cite{aco} for their energy.
Both results giving an error $o(\l^2)$ for the energy and $O(\l)$ for the number of points seem optimal: for the energy since  $W$ is itself defined as a limit over large size domains, and for the number of points, since a lattice configuration would give rise to the same error (note that  we are counting the error made on squares and not on balls, the latter one would be smaller -- more precisely in $R^{2/3}$ --  for a lattice). 
Recall finally that we cannot characterize $W(E_\vp, \1_{K_\l(a)})$ because of the possibility of points being on $\p K_\l(a)$ which is why we resort to $W(E_\vp, \chi_{K_\l(a)})$ instead. 

Note that the results \eqref{eqnumpointscb} and \eqref{eqnumpointsper} imply that away from the boundary
the distance from each point to its nearest neighbor is bounded above by some constant depending only on  $m$ (in fact scaling the problem to reduce to  $m=1$, one can see that it is $\frac{C}{\sqrt{m}}$, for some universal $C$). The opposite inequality, i.e. that the distance from each point to its nearest neighbor is bounded below by some $\frac{c}{\sqrt{m}}$ for $c>0$ universal can be obtained (at least in the periodic case, or away from the boundary) by the same argument we will give below for the Coulomb gas, due to E. Lieb \cite{lieb}.

\subsection{Main results on the Coulomb gas} \label{sec1.3}
We next turn to the Coulomb gas, more precisely to the analysis of minimizers of \eqref{wn}.

 We mentioned that \eqref{eqcvminim} holds, where $\mu_0$ is the unique minimizer over probability measures of the ``mean-field energy"
 \begin{equation}
 I(\mu)= - \iint_{\mr^2 \times \mr^2} \log |x-y| \, d\mu(x)\, d\mu(y) + \int_{\mr^2} V(x) \, d\mu(x)\end{equation}
 called the ``equilibrium measure" (see \cite[Chap. 1]{safftotik} and references therein).  We denote $\E = \supp (\mu_0)$.
 We use here the same assumptions as in \cite{sandierserfatygas} i.e.
 \begin{equation}\label{assumpV1}  \lim_{|x|\to +\infty} \frac{V(x)}{2} - \log|x| = +\infty, \end{equation}
\begin{equation}\label{assumpV2} \text{ $V$ is $\mathcal C^3$ and there exists $\bw,\wb>0$  s.t. $\bw\le \frac{\Delta V}{4\pi} \le \wb$,}\end{equation}
\begin{equation}\label{assumpV3} \text{$V$ is such that $\p \E$ is $\mathcal C^{1}$.}    \end{equation}
The assumption \eqref{assumpV1} ensures in particular that $I$ has a minimizer, which has compact support $\E$.  Also $\mu_0$ has a density $m_0$: $d\mu_0(x)=m_0(x)dx$ and $m_0=\frac{\Delta V}{4}\indic_{\E}.$ Hence, in view of \eqref{assumpV2} we have
\begin{equation}\label{minom}
0<\bw \le m_0 \le \wb,
\end{equation}
and
\begin{equation}
|\nab m_0|\le C.
\end{equation}

Because there are $n$ points in $\E$, a set of  fixed size, it is natural to blow-up everything to the scale $\sqrt{n}$, in order to obtain points that are separated by order $1$ distances. We will thus denote $x_i'= \sqrt{n} x_i$, but also $\E'=\sqrt{n} \E$ and $m_0'(x')= m_0(x'/\sqrt{n})$ the blown-up density of $\mu_0$, and $d\mu_0'= m_0' (x') dx'$. We will also write 
\begin{equation}\label{defnu'}
\nu_n'= \sum_{i=1}^n \delta_{x_i'}.\end{equation}

In \cite{sandierserfatygas} we studied minimizers of $w_n$ via $\Gamma$-convergence. More precisely we showed that $\frac{1}{n} ( w_n - n^2 I(\mu_0)+ \frac{n}{2}\log n)
$ $\Gamma$-converges to some suitable average of $W$ (as in Definition \ref{defrenenergy}), computed over blow-ups of configurations of points.
The starting point, that we will need here, is the exact ``splitting formula" of  \cite[Lemma 2.1]{sandierserfatygas}: for any $n$ and any  $x_1, \dots, x_n \in \mr^2$
\begin{equation}\label{splitting}
w_n(x_1, \dots, x_n) = n^2 I(\mu_0)- \frac{n}{2}\log n + 2n\sum_{i=1}^n \zeta(x_i)+
\frac{1}{\pi} W(E_n, \indic_{\mr^2}).\end{equation}
Here $\zeta$ is a fixed function (depending only on $V$),  given more precisely by $- \log * \mu_0 + \hal V - cst$ (see \cite{sandierserfatygas} for details), which satisfies 
$$\zeta \ge 0 \qquad \{\zeta = 0\} = \E.$$
  The vector field  $E_n$ is defined as the ``electric field" generated by the blown-up points and their background:
\begin{equation}\label{En}
E_n= 2\pi \nabla \Delta^{-1} \( \sum_{i=1}^n \delta_{x_i'}- m_0'(x') \)\end{equation}
where $2\pi \Delta^{-1}$ is the operator of convolution by $ -\log $, and the quantity $W(E_n, \indic_{\mr^2})$ is defined according to \eqref{defWchi}.
We then note that, if all the points are in $\E$ then $\sum_i \zeta(x_i)=0$, and so   being a minimizer of $w_n$ is equivalent to being a minimizer of $W(E_n,\indic_{\mr^2})$, in view of \eqref{splitting}. We will first show that it is the case for minimizers of $w_n$, and thus reduce ourselves to studying  minimizers of $W(E_n,\indic_{\mr^2})$. This is how the  analysis used for proving Theorems \ref{thmmaincb} and \ref{thmmainper} is then exactly suited. The main difference here is that we have to deal with a varying background (this works fine in the end because the background $m_0'$ is very slowly varying as $n\to \infty$).

After the splitting of \eqref{splitting}, the approach in \cite{sandierserfatygas} consisted in obtaining a general lower bound for $\frac{1}{n\pi} W(E_n, \indic_{\mr^2})$ in the limit $n\to \infty$, as well as a matching upper bound.   The conclusions were appropriate to ``almost minimizers", i.e. configurations whose energy is asymptotically the minimal energy. As a result they were weaker and they were averaged: they said that almost minimizers are such that blown-up configurations tend to minimize $W$, for ``almost every blow-up center in $\E$". It is normal that almost minimizers should admit the possibility of the energy being larger than expected on a set of asymptotically vanishing volume fraction. 

Here we work rather on true  minimizers, and obtain stronger results. This will use the local  minimality  of a minimizer (with respect to its own boundary condition) as well  the other  specific property of minimizers that all their points are in $\E$.
\begin{thm}\label{th3} Let $(x_1,\dots,x_n)$ be a minimizer of $w_n$. Let  $m_0$, $m'_0$ and $\j_n$ be as above.
The following holds:
\begin{enumerate}
	\item\label{itempointth3} for all $i\in [1,n]$, $x_i\in \Sigma$;
	\item\label{itemequidistth3} there exist $\beta \in (0,1)$, $\bar c>0$, $C>0$ (depending only on $\|m_0\|_{L^\infty} $),  such that for every $\l\ge \bar c$ and $a\in \E'$ such that $d (K_\l(a), \p \E') \ge  n^{\beta/2}$, we have
	\begin{equation}\label{137}
	\limsup_{n\to \infty}\frac{1}{\l^2}\left\vb W(\j_n,\chi_{K_{\l}(a)} )-\int_{K_\l(a)}\Big(\min_{\mathcal A_{m'_0(x)}}W\Big)\,dx\right\vb\le o(1)_{\l\to+\infty}.
	\end{equation}
	and
	\begin{equation}\label{eqnumpointsgas}
\limsup_{n\to \infty}	\left\vb \nu_n'(K_\l(a))-\int_{K_{\l}(a)}m'_0(x)\,dx\right\vb\le C\l,
	\end{equation}where $\nu_n'$ is defined in \eqref{defnu'}.
\end{enumerate}
\end{thm}As already  mentioned,  \eqref{eqnumpointsgas} should be optimal and improves the result of \cite{aoc} but with a slightly stronger restriction on the distance to $\p \E$; while \eqref{137} says more, since it  says that minimizers have to  behave  at the microscopic scale  like minimizers of $W$ in the appropriate class $\mathcal A_{m_0'}$, and so can  be expected to look like Abrikosov triangular lattices (this is only a  conjecture of course). The proof of \eqref{eqnumpointsgas} is in fact derived from \eqref{137}.

The following result is adapted from the unpublished  result of Lieb  \cite{lieb} in the case of a constant background.
\begin{thm}[\cite{lieb}]\label{thlieb}
Let $(x_1,\dots,x_n)$ be a minimizer of $w_n$, and $x_i'= \sqrt{n} x_i$.
Then there exists $r_0>0$ depending only on  $\|m_0\|_{L^\infty}$ (hence on $V$) such that $$\min_{i\neq j} |x_i'-x_j'| \ge r_0.$$
\end{thm}

On the other hand, \eqref{eqnumpointsgas} also easily implies, in view of \eqref{minom},  that there exists some $R_0>0$ (depending only on $\bw$ in \eqref{minom}) such that, if $d(x_i', \E') \ge n^{\beta/2}$,    the distance from any $x_i'$ to its nearest neighbor is bounded by $R_0$.
Combined with Theorem \ref{thlieb}, this establishes the following
\begin{cor}
\label{corodist}
 Let $(x_1,\dots,x_n)$ be a minimizer of $w_n$. Let  $m_0$, $m'_0$ and $\j_n$ be as above.
Then there exists $r_0>0$ and $R_0>0$  depending only on  $\mu_0$ (hence on $V$) such that,  if $d(x_i, \E) \ge n^{\frac{\beta-1}{2}}$, the distance from $x_i$ to its nearest neighbor    is in $[r_0/\sqrt{n}, R_0/\sqrt{n}]$.
\end{cor}

\subsection{Open questions and plan of the paper}
Let us conclude by a set of open questions that should be solvable by the methods we used here:
\begin{itemize}
\item One could most likely adapt our method here to prove similar equidistribution results of points and energy for the one-dimensional renormalized energy introduced in \cite{ss1d} and for minimizers of  the one-dimensional log gases energy also studied in \cite{ss1d}.
\item In \cite{rs}, another renormalized energy is introduced, and extracted as a limiting energy for Coulomb gases in any dimension $d\ge 2$. There would remain to prove the same results as here for minimizers.
\item
Finally,  the methods used here should allow in principle (although the setting is significantly more complex) to obtain similar equidistribution results as Theorem \ref{th3} for the energy and the vortices of minimizers of the Ginzburg-Landau energy, thus improving again on the  averaged results obtained via $\Gamma$-convergence in   \cite{sandierserfaty}. More precisely this would mean, since in that context,  the equivalent of the mean-field limit measure $\mu_0$ is constant on its support $\Sigma$, that the density of vortices is uniform on $\Sigma$ on all scales much larger than the intervortex distance, and that the Ginzburg-Landau energy density on  such scales is also constant and asymptotic to $\min W$.
\end{itemize}

The paper and the proofs are organized as follows. All our proofs are energy-comparison based. They rely on the techniques introduced in \cite{sandierserfaty,sandierserfatygas}, in particular the fact that even though $W(\j, \chi)$ is an energy with a singular density, it can be replaced by a energy density which is bounded below, once the number of points near the boundary is well controlled. This is the content of Proposition \ref{prop49ss} which is recalled in Section \ref{sec3}, and allows to control $\j$ in $L^p$ for $p<2$ via $W(\j, \chi)$, see Lemma \ref{lemnormp}.
 
In Section \ref{sec2}, we start with the case of a constant background equal to $m$. By using extension lemmas, we show that $\sigma_0, \sigma_\vp $ and $\sigma_{per} $ are all asymptotically equal to $\sigma_m^* $ in the limit of large squares or tori.
In Section \ref{sec3} we prove Theorems \ref{thmmaincb} and \ref{thmmainper} by using a bootstrap argument: given a minimizer in a large box,  we show that by a mean-value argument we can find a much smaller box (but not too small either) with a good boundary trace. Using then that a minimizer is also a minimizer on any smaller box with respect to its own boundary data, and combining with the results of Section \ref{sec2} we deduce the value of the energy on the smaller box. A bootstrap argument is then used to go down to any small size box (as long as its size $\l$ is still bigger than some constant). Finally we show how controlling $W$ down to $O(1)$ scales allows to deduce \eqref{eqnumpointscb} or \eqref{eqnumpointsper}.

In Section \ref{sectionnonconstant} we turn to the case of  a varying background and show how to adapt similar  results  to Section \ref{sec2}, with error terms depending explicitly  on the oscillation of the background. 
In Section \ref{sec5} we turn to the Coulomb gas minimizers and prove that all their points lie in $\E$ as well as  Theorem \ref{thlieb} (these both rely on similar arguments, totally independent from the rest of the paper), then we conclude with the proof of Theorem \ref{th3} by adapting the ideas of Sections \ref{sec2}, \ref{sec3}, \ref{sectionnonconstant}. 
Finally, some technical results, mostly adapted from \cite{sandierserfaty}, are gathered in the appendix. 

\vskip .1cm
{\bf Acknowledgements: } Part of this work originates in the Master's thesis of H. Ben Moussa directed by Etienne Sandier and the second author.
We are grateful to Prof. Lieb for providing the idea of the  proof of Theorem \ref{thlieb} and allowing us to reproduce it here. We also would like to thank Etienne Sandier for many useful discussions. 
The research of both authors was  supported by a EURYI award. Moreover, the research of the first author was partially supported by the Grant ANR-10-BLAN 0101. We also thank the Forschungsinstitut f\"ur Mathematik at the ETH Z\"urich, where part of this work was completed, for its hospitality.

\section{Comparison of different minimization problems}\label{sec2}

In this section, we prove that under some assumptions on the function $\varphi$, the quantities $\sigma_\varphi(K_L(a); m)$ and $\sigma_\varphi(K_L(a); \rho)$ are well defined. Moreover, we compare the different minimization problems defined above. The results that we will prove in this section are the following.

\begin{prop}
	\label{existencemin} Let $p\in(1,2)$, $a\in \RR^2$ and $L>0$ be fixed. Let $\varphi\in L^p(\partial K_L(a))$ such that \eqref{h1boundary} and \eqref{h2boundary} are satisfied. If $\rho\in L^{\infty}(K_L(a))$, then the minimum of $W(E,\1_{K_L(a)})$ over $\mathcal B_{\rho,\varphi}(K_L(a))$ is achieved.
\end{prop}

\begin{prop} \label{proplimmincon} Let $p\in(1,2)$ and $m$ be a positive number.
Then:
\begin{enumerate}
\item for all sequences of real numbers $L$ such that $m\vb K_L(a)\vb \in \NN$, we have
\begin{equation}\label{eqlimsigma0con}
\lim\limits_{L\to+\infty} \sigma_{0}(K_L(a);{m})=\sigma_{m}^*;
\end{equation}
\item for all sequences of real numbers $L$ such that $m\vb \mathbb T_L\vb \in \NN$, we have
\begin{equation}\label{eqlimsigmaperio}
\lim\limits_{L\to+\infty} \sigma_{per}(L;m)=\sigma_m^*.
\end{equation}
\item  given $\gamma\in\left(\frac{3-p}{2},1\right)$ and $M>0$, we have
\begin{equation}\label{eqlimsigmaphicon}
\lim\limits_{L\to+\infty} \sigma_{\varphi}(K_L(a);m)=\sigma_{m}^*.
\end{equation}
 uniformly w.r.t. $\varphi$ such that \eqref{h1boundary} and \eqref{h2boundary} are satisfied in $K_L(a)$.

\end{enumerate}
\end{prop}

We start with the following ``screening" proposition adapted from the ideas of \cite{sandierserfaty}. The first part allows to extend a given configuration with given boundary data satisfying \eqref{h2boundary} to a strip outside of a square,  bringing the boundary value to $0$, while keeping the number of points and the energy controlled (what matters is that they remain negliglible compared to the volume of the square). The second part allows to do the same inside a square.

In this proof, as well as in all the sequel, several lengthscales will appear: the lengthscale $\l$ of a given square, the lengthscale $\l^\gamma\ll \l $ which is the width needed to obtain good boundaries satisfying \eqref{h2boundary} by mean-value arguments, and the lengthscale $\l^\a  \ll \l^\gamma$ which is the width needed to transition from a boundary data with \eqref{h2boundary} to a zero boundary data.
These exponents will remain the same in the whole paper.

We will most often  assume that the background densities $\ro$ are in $\mathcal C^{0,\lambda}$ with $\hal \le \lambda \le 1$.

\begin{prop}\label{propj+j-} Let $p\in(1,2)$, $\gamma\in\left(\frac{3-p}{2},1\right)$, $\lambda \in [\hal, 1]$,  $M$ a positive constant and $a\in\RR^2$. Let $\varphi\in L^p(\partial K_\l(a))$ satisfying \eqref{h2boundary} on $\p K_\l(a)$. There exist constants $1<\a<\gamma$ and $\beta\in (0,1)$, depending on $p$, $\gamma$, $\lambda$,  
 for which the following holds.
\begin{enumerate}[label=\emph{\arabic*.}]
\item
 Let  $\rho$ be a $\mathcal C^{0,\lambda}(K_{\l+2\l^\a}(a)\backslash K_{\l}(a))$ function for which there exist $\underline \rho, \overline \rho>0$ such that $\underline \rho \le \rho(x)\le \overline \rho$.
 There exist $C,c$ positive constants depending only on $p$, $\gamma$, $M$, $\underline \rho$ and $\overline \rho$,
 such that for all $\l\ge c$ the following holds:

\label{itemj+} There exist $t_+\in[\l+\l^\a,\l+2\l^\a]$ and $\j_+:\mathcal K_+^a\to \RR^2$ with $\mathcal K_+^a=K_{t_+}(a)\backslash K_\l(a)$ such that
	\begin{equation}
		\label{defj+}
		\left\{
		\begin{aligned}
			&\di \j_+=2\pi \Big(\sum_{p\in\Lambda_+}\delta_p-\rho\Big)&&\quad\text{in}\ \mathcal K_+^a\\
			&\j_+\cdot\vnu=0&&\quad{on}\ \partial K_{t_+}(a)\\
			&\j_+\cdot\vnu=\varphi&&\quad\text{on}\ \partial K_{\l}(a)
		\end{aligned}
		\right.
	\end{equation}
	where $\Lambda_+$ is a discrete subset of the interior of $\mathcal K_+^a$,  with
 \begin{equation}
 \label{eqnumpointsj+} \# \Lambda_+\le C \l^{1+\a},\end{equation}
  whose elements have distances, and distances to the boundary of $\mathcal K_+^a$,  all bounded below by a constant depending only on $\overline \ro$, and it holds that $\curl \j_+=0 $ in a neighborhood of these points; moreover, we have
	\begin{align}
		\label{energyj+}
		W(\j_+,\1_{\mathcal K^a_+})\le&\, C \l^{1+\beta}+ C \l^2 \left(\norm \rho\norm_{C^{0,\lambda}(\mathcal K^a_+)}\l^{\beta\lambda}\right)\left(1+\norm \rho\norm_{C^{0,\lambda}(\mathcal K^a_+)}\l^{\beta\lambda}\right),
	\end{align}
	and for any nonnegative  function $\chi\le 1$ such that $\chi=1$ on $\partial K_{l}(a)$, $\chi=0$ on $\partial K_{t_+}(a)$ and $\vb \nabla \chi\vb$ is bounded, we have
	\begin{align}
		\label{energyj+chi}
		W(\j_+,\chi)\le&\, C( \l^{1+\beta}+\l^{1+\a})+ C \l^2 \left(\norm \rho\norm_{C^{0,\lambda}(\mathcal K^a_+)}\l^{\beta\lambda}\right)\left(1+\norm \rho\norm_{C^{0,\lambda}(\mathcal K^a_+)}\l^{\beta\lambda}\right).
	\end{align}
\item\label{itemj-} 
Let $\frac{1}{2}\le\lambda\le 1$ and $\rho$ be a $\mathcal C^{0,\lambda}(K_{\l}(a)\backslash K_{\l- 2\l^{\a}}(a))$ function for which there exist $\underline \rho, \overline \rho>0$ such that $\underline \rho \le \rho(x)\le \overline \rho$.
 There exist $C,c$ positive constants  depending only on $p$, $\gamma$, $M$, $\underline \rho$ and $\overline \rho$,
such that for all $\l\ge c$ the following holds:

There exist $t_-\in[\l-2\l^\a,\l-\l^\a]$ and $\j_-:\mathcal K_-^a\to \RR^2$ with $\mathcal K_-^a=K_{\l}(a)\backslash K_{t_-}(a)$ such that
	\begin{equation}
		\label{defj-}
		\left\{
		\begin{aligned}
			&\di \j_-=2\pi \Big(\sum_{p\in\Lambda_-}\delta_p-\rho\Big)&&\quad\text{in}\ \mathcal K_-^a\\
			&\j_-\cdot\vnu=0&&\quad{on}\ \partial K_{t_-}(a)\\
			&\j_-\cdot\vnu=\varphi&&\quad\text{on}\ \partial K_{\l}(a)
		\end{aligned}
		\right.
	\end{equation}
	where $\Lambda_-$ is a discrete subset of the interior of $\mathcal K_-^a$ satisfying the same properties as $\Lambda_+$, among which
\begin{equation}
\label{eqnumpointj-}
		\# \Lambda_- \le  C\l^{1+\a},\end{equation}
	\begin{align}
		\label{energyj-}
		W(\j_-,\1_{\mathcal K^a_-})\le&\, C \l^{1+\beta}+ C \l^2 \left(\norm \rho\norm_{C^{0,\lambda}(\mathcal K^a_-)}\l^{\beta\lambda}\right)\left(1+\norm \rho\norm_{C^{0,\lambda}(\mathcal K^a_-)}\l^{\beta\lambda}\right),
	\end{align}and for any nonnegative function $\chi \le 1$ such that $\chi=1$ on $\partial K_{t-}(a)$, $\chi=0$ on $\partial K_{l}(a)$ and $\vb \nabla \chi\vb$ is bounded, we have
	\begin{align}
		\label{energyj-chi}
		W(\j_-,\chi)\le&\, C ( \l^{1+\beta} +\l^{1+\a} ) + C \l^2 \left(\norm \rho\norm_{C^{0,\lambda}(\mathcal K^a_-)}\l^{\beta\lambda}\right)\left(1+\norm \rho\norm_{C^{0,\lambda}(\mathcal K^a_-)}\l^{\beta\lambda}\right).
	\end{align}
\end{enumerate}
\
\end{prop}

		
The proof of this proposition uses a few contruction lemmas, mostly adapted from \cite{sandierserfaty} to a nonconstant background, which we now state and whose proof is in the appendix.

\begin{lem}\label{lemsysum} Let $\mathcal{R}$ be a rectangle with sidelengths in $\left[\frac{L}{2},\frac{3L}{2}\right]$. Let $p\in (1,2)$. Let $\varphi\in L^p(\partial\mathcal{R})$ be a function which is $0$ except on one side of the rectangle $\mathcal{R}$. Let  $\rho$ be a nonnegative $\mathcal C^{0}$  function.
Let $m$ be a constant such that $(m-\rho_{\mathcal R})|\mathcal{R}|=-\frac{1}{2\pi}\int_{\partial\mathcal{R}}\varphi$ where $\rho_{\mathcal R}:=\fint_{\mathcal R}\rho(x)\,dx$. Then the mean zero solution to
\begin{equation}\label{eqsysurho}
\left\{\begin{aligned}
-\Delta u &=2\pi (m-\rho(x))& &\mbox{in}& &\mathcal{R}\\
\nabla u \cdot \vnu&=\varphi & & \mbox{on}& &\partial\mathcal{R}
\end{aligned}\right.
\end{equation}
satisfies for every $q\in [1,2p]$
\begin{equation}\label{eqnormurho}
\int_{\mathcal{R}}|\nabla u|^q\le C_{p,q}L^{2-\frac{q}{p}}\|\varphi\|^q_{L^p(\partial\mathcal{R})}+CL^{q+2}\norm \rho-\rho_\mathcal{R}\norm_{L^{\infty}(\mathcal R)}^q.
\end{equation}
\end{lem}
\begin{lem}[\cite{sandierserfaty}]\label{lemsysfdelta} Let  $m$ be a positive constant. Let $\mathcal R$ be a rectangle of barycenter $0$, sidelengths in $\sqrt{\frac{1}{m}}\left[\frac{1}{2},\frac{3}{2}\right]$, and such that $m|\mathcal{R}|=1$. Then the solution to
\begin{equation}\label{eqsysfdelta}
\left\{\begin{aligned}
-\Delta f &=2\pi(\delta_0-m)& &\mbox{in}& &\mathcal{R}\\
\nabla f \cdot  \vnu&=0 & & \mbox{on}& &\partial\mathcal{R}
\end{aligned}\right.
\end{equation}
satisfies
\begin{equation}\label{eqenergyf}
\lim_{\eta\to0}\left\vb\int_{\mathcal{R}\backslash B(0,\eta)}|\nabla f|^2+2\pi\log\eta\right\vb\le C
\end{equation}
where $C$ 
 depends only on $m$, and for every $1\le q<2$
\begin{equation}\label{eqnormf}
\int_{\mathcal{R}}|\nabla f|^q\le C_q  m^{\frac{q}{2}-1},
\end{equation}
where $C_q$ depends only on $q$.
\end{lem}

We now give the proof of Proposition \ref{propj+j-}. The proof is obtained by arguments similar to those of \cite[Proposition 4.2]{sandierserfaty} but we need to keep more carefully track of the errors and the exponents. Moreover, we have to deal with the fact that the background is a non-constant function $\rho$.

\begin{proof}[Proof of Proposition \ref{propj+j-}]
Without loss of generality, we may assume $a=0$ and $K_\l= [-\l,\l]^2$. 
We construct $\j_-$ on $K_\l\backslash K_{t_-}$; the proof for $\j_+$ is \smallskip similar.

\noindent\textbf{Step 1.} We create a layer of width $\l^\alpha$ on which we connect to a zero boundary data.
\\
  Let $\alpha\in\left(\frac{2-\gamma}{1+p},\frac{p-2+\gamma}{p-1}\right)$; such an $\alpha$ exists because $\gamma>\frac{3-p}{2}$. Since $p<2$ and $\gamma<1$, we have $\alpha<\gamma$.
We start by building a strip near the lower part of $\p K_\l$ which is compatible with the desired boundary data.
For $t>0$, denote $S_t= [-\l,\l]\times [-\l, -\l + t]$.  We claim there exists $t\in [\hal \l^\alpha,\l^\alpha]$ such that 
\begin{equation}\label{claimf}
f(t):=\int_{S_t} \ro(x)\, dx+ \frac{1}{2\pi} \int_{\p K_\l \cap \p S_t}\vp \in \mn^*.
\end{equation}
By H\"older's inequality and \eqref{h2boundary}, we have, for any $t \in [\hal \l^\alpha, \l^\alpha]$, 
$$\left|   \int_{\p K_\l \cap \p S_t}\vp \right|\le M^{\frac{1}{p}}\l^{\frac{2-\gamma}{p}}\l^{1-\frac{1}{p}}.$$
Therefore,  
$$f(\hal \l^\alpha)= \int_{S_{\hal \l^\alpha}}\ro(x)\, dx + \frac{1}{2\pi} \int_{\p K_\l \cap \p S_{\hal \l^\alpha}}\vp 
\ge \underline  \ro  \l^{1+\alpha} -  \frac{1}{2\pi} M^{\frac{1}{p}}\l^{1+ \frac{1-\gamma}{p}}\ge 0 $$
for $\l$ large enough (depending on $p, \gamma, M, \underline{\ro}$), since $\alpha >\frac{1-\gamma}{p}$. 
It follows that 
\begin{equation*}f(t) \ge f(\hal \l^\alpha) +  \underline \ro  (t- \hal \l^\alpha) \l -    M^{\frac{1}{p}}\l^{1+ \frac{1-\gamma}{p}}
\\ \ge  \underline \ro  (t- \hal \l^\alpha) \l - \frac{1}{2\pi}   M^{\frac{1}{p}}\l^{1+ \frac{1-\gamma}{p}}.
\end{equation*}
Thus $f(\l^\alpha)>1$ again for $\l $ large enough.
By a mean-value argument, since $f(t)$ is continuous, we find that there exists $t \in [\hal\l^\alpha, \l^\alpha]$ for which \eqref{claimf} holds.

Let $t$ be that value. We next split the strip $S_t$ into a finite number of rectangles $\mathcal R_i$ of width $\in [\hal \l^\alpha, \l^\alpha]$.
This follows the same reasoning: we claim there exists $s\in [\hal \l^\alpha, \l^\alpha]$ such that 
\begin{equation}\label{claimh}
 h(s):= \int_{[-\l, -\l + s]\times [-\l, -\l +t]} \ro(x)\, dx+ \frac{1}{2\pi} \int_{\p K_\l \cap \p ([-\l, -\l + s]\times [-\l, -\l +t]) }\vp \in \mn^*.
\end{equation}
By H\"older's inequality and \eqref{h2boundary} we have 
$$h(\hal \l^\alpha) \ge  \frac{1}{4}\underline \ro \l^{2\alpha} - C M^{\frac{1}{p}}\l^{\frac{2-\gamma}{p}} \l^{\alpha(1-\frac{1}{p})} \ge 0,$$
when $\ell $ is large enough. For that it suffices to check that $\alpha\left(1-\frac{1}{p}\right)+\frac{2-\gamma}{p}<2\alpha$, which is true by choice of $\alpha$.
Arguing in the same way as above, we find that $h(\l^\alpha)- h(\hal \l^\alpha)>1$ and thus by a mean value argument there exists $s\in [\hal \l^\alpha, \l^\alpha]$ such that \eqref{claimh} holds.
We define the first rectangle $\mathcal R_1$ to be $[ -\l, -\l + s]\times [-\l, -\l +t]$. We may then iterate this reasoning to build a rectangle $\mathcal R_2 $ of the form 
$[-\l +s, \l+ s+r] \times [-\l, \l +t]$ for some $r \in [\hal \l^\alpha, \l^\alpha]$, etc, until the whole strip $S_t$ is exhausted. If the last rectangle is too narrow, we may merge it with the one before last, and this ensures a collection $\{\mathcal R_i\}$ of rectangles of sidelengths in $[\hal \l^\alpha, \frac{3}{2}\l^\alpha]$ partitioning $S_t$  and such that 
\begin{equation}\label{quantirect}
 \int_{\mathcal R_i } \ro(x)\, dx+ \frac{1}{2\pi} \int_{\p K_\l \cap \p\mathcal R_i }\vp \in \mn^*.\end{equation}

The construction can then be repeated near the other three sides of $K_\l$: we may find three disjoint strips in $K_\l\backslash S_t$, which we can each again split into rectangles $\mathcal R_i$ on which \eqref{quantirect} holds. 
We still denote by $\{\mathcal R_i\}$ the total collection of rectangles and we observe that  $K_\l\backslash \cup_i \mathcal R_i$ is a \smallskip rectangle.

\noindent{\bf Step 2.} We define $\j_-$ in each $\mathcal R_i$.
We let $\varphi_i$ denotes the restriction of $\varphi$ to $\p \mathcal R_i \cap \p K_\l$ extended by $0$ on the rest of $\partial\mathcal{R}_i$.
The condition \eqref{quantirect}  means that 
  \begin{equation}\label{eqcondRi}
\int_{\mathcal{R}_i} \rho(x)\,dx+\frac{1}{2\pi}\int_{\partial \mathcal R_i}\varphi_i\in \NN.
\end{equation}
We also let 
\begin{equation}\label{defrho_i}
\tilde \rho_i= \rho_{\mathcal R_i}+\frac{\int_{\partial \mathcal R_i}\varphi_i}{2\pi |\mathcal{R}_i|}
\end{equation}
where we denote, for any rectangle $\mathcal R$, $\rho_{\mathcal R}=\fint_{\mathcal R}\rho(x)\,dx$. Using H\"older's inequality and \eqref{h2boundary} again, we have
\begin{align*}
\vb \tilde \rho_i-\rho_{\mathcal R_i} \vb\le C \l^{\alpha\left(-1-\frac{1}{p}\right)}\left(\int_{\partial K_\l}\vb \varphi\vb^p\right)^{1/p}\le C_M \l^{\alpha(1-\frac{1}{p} ) + \frac{2-\gamma}{p} }.    
\end{align*}
Since $\alpha>\frac{2-\gamma}{1+p}$, we deduce that $\vb\tilde  \rho_i- \rho_{\mathcal R_i}  \vb=o(1)_{\l\to +\infty}$.

 By  (\ref{eqcondRi}) and \eqref{defrho_i}, we have $ \tilde \rho_i\vb\mathcal R_i\vb\in\mathbb N$, and since $\tilde \rho_i$ is equivalent to $\rho_{\mathcal R_i}$, this integer number belongs to $[\hal \underline \rho |\mathcal R_i|, 2\overline\ro |\mathcal R_i|]$.  We may then partition $\mathcal R_i$ into $ \tilde \rho_i\vb\mathcal R_i\vb$  rectangles $\mathcal R_{ik}$, whose sidelengths are in $\sqrt{\frac{1}{\tilde \rho_i}}\left[\frac{1}{2},\frac{3}{2}\right]$ and such that for each $i,k$, we have $\tilde \rho_i\vb\mathcal R_{ik}\vb=1$. 
 
 On each of these rectangles, we apply Lemma \ref{lemsysfdelta} with $m=\tilde\rho_i$, this yields a function $f_{ik}$ satisfying (\ref{eqenergyf}) and (\ref{eqnormf}). We then define the vector field $\j_1$ in $\mathcal G:=\cup_i\mathcal R_i$ by $\j_1=-\nabla f_{ik}$ in each $\mathcal R_{ik}$. Since no divergence is created at the interface between the $\mathcal R_{ik}$ (because $\nab f_{ik} \cdot \vnu=0$), we obtain
\begin{equation}\label{eqsysj1}
\left\{
\begin{aligned}
&\di \j_1=2\pi\Big(\sum_{p\in \Gamma}\delta_p-\sum_i \tilde \rho_i \1_{\mathcal R_i}\Big)&&\mbox{in }\mathcal G\\
&\j_1\cdot\vnu=0&&\mbox{on }\partial \mathcal G
\end{aligned}
\right.
\end{equation}
where $\Gamma$ is the union over $i,k$ of the centers of the rectangles $\mathcal R_{ik}$.
Moreover, by (\ref{eqenergyf}) and (\ref{eqnormf}), $\j_1$ satisfies
\begin{equation}\label{eqenergyj1}
\lim_{\eta\to0}\left\vb\frac{1}{2}\int_{\mathcal{G}\backslash \cup B(p,\eta)}|\j_1|^2+\pi\#\Gamma\log\eta\right\vb\le C \l^{1+\alpha}
\end{equation}
and for $q<2$
\begin{equation}\label{eqnormj1}
\int_{\mathcal{G}}|\j_1|^q\le C_q \l^{1+\alpha},
\end{equation}
since the number of $\mathcal R_{ik}$ is of order $\vb\mathcal R_i\vb=O(\l^{2\alpha})$ for each $i$, and the number of rectangles $\mathcal R_{i}$ is $O(\l^{1-\alpha})$. The constant $C$ depends only on $p,\gamma, M, \underline \rho$ and $\overline \rho$; while $C_q$ depends also on $q$.

Next, since $( \tilde \rho_i- \rho_{\mathcal R_i})\vb\mathcal R_i\vb=\frac{1}{2\pi}\int_{\partial \mathcal R_i}\varphi_i$, we may apply Lemma \ref{lemsysum} in each $\mathcal R_i$ with $\tilde g=-\varphi_i$ for boundary data. Then we define the vector field $\j_2=-\nabla u_i$ which satisfies
\begin{equation}\label{eqsysj2}
\left\{
\begin{aligned}
&\di \j_2=2\pi \Big(\sum_i \tilde \rho_i \1_{\mathcal R_i}-  \rho\Big)&&\mbox{in }\mathcal G\\
&\j_2\cdot\vnu=-\tilde g&&\mbox{on }\partial \mathcal G
\end{aligned}
\right.
\end{equation}
(again no divergence is created at the interface).
We recall that $\tilde g=-\varphi$ on $\partial K_\l$ and $0$ on the rest of $\partial\mathcal G$. Moreover, for every $q\in[1,2p]$, we have
\begin{equation*}
\int_{\mathcal{R}_i}|\j_2|^q\le C_{p,q}\l^{\alpha\left({2-\frac{q}{p}}\right)}\|\varphi_i\|^q_{L^p(\partial\mathcal{R}_i)}+C\l^{\alpha(q+2)}\norm \rho-\rho_{\mathcal{R}_i}\norm_{L^{\infty}(\mathcal R_i)}^q.
\end{equation*}
Adding these relations, we obtain
\begin{align*}
\int_{\mathcal{G}}|\j_2|^q\le& C_{p,q}\l^{\alpha\left({2-\frac{q}{p}}\right)}\left(\sum_i\|\varphi_i\|^q_{L^p(\partial\mathcal{R}_i)}\right)+C\l^{\alpha(q+2)}\left(\sum_i \norm \rho-\rho_{\mathcal{R}_i}\norm_{L^{\infty}(\mathcal R_i)}^q \right).
\end{align*}
Since $\frac{q}{p}>1$ and the number of rectangles in $\mathcal G$ is $O(\l^{1-\alpha})$, we deduce
\begin{equation*}
\int_{\mathcal{G}}|\j_2|^q\le C_{p,q}\l^{\alpha\left({2-\frac{q}{p}}\right)}\left(\int_{\partial K_\l}\vb \varphi\vb^p +\l^{1-\alpha}\right)^{\frac{q}{p}}+C\l^{\alpha(q+2)}\left(\sum_i \norm \rho-\rho_{\mathcal{R}_i}\norm_{L^{\infty}(\mathcal R_i)}^q \right).
\end{equation*}
Then, using the hypothesis (\ref{h2boundary}) and the fact that $\norm \rho-\rho_{\mathcal{R}_i}\norm_{L^{\infty}(\mathcal R_i)}\le C\norm \rho\norm_{\mathcal C^{0,\lambda}}\l^{\alpha\lambda}$, we have
\begin{equation}\label{eqnormj2}
\int_{\mathcal{G}}|\j_2|^q\le C_{p,q}\l^{\alpha\left({2-\frac{q}{p}}\right)+\frac{q}{p}(2-\gamma)}+C\l^{\alpha q+\alpha +1}\norm \rho\norm_{\mathcal C^{0,\lambda}}^q\l^{\alpha q\lambda }
\end{equation}
with $C_{p,q}$ a positive constant that depends on $p,q$ and $M$, and $C$ a universal constant. Moreover for all $q<\frac{2p(1-\alpha)}{2-\alpha-\gamma}$, we have $\alpha\left({2-\frac{q}{p}}\right)+\frac{q}{p}(2-\gamma)<2$. We remark that since $\alpha\in\left(\frac{2-\gamma}{1+p},\frac{p-2+\gamma}{p-1}\right)$, we have $\frac{2p(1-\alpha)}{2-\alpha-\gamma}>2$. Thus, we may find some $q>2$ such that
\begin{equation}\label{eqnormj2bis}
\int_{\mathcal{G}}|\j_2|^q\le C_{p,q}\l^{\sigma}+C\l^{\alpha q+\alpha +1}\left(\norm \rho\norm_{\mathcal C^{0,\lambda}}\l^{\alpha\lambda}\right)^q
\end{equation}
for some $\sigma<2$.

Finally, we define $\j_-=\j_1+\j_2$ and $\nu=\sum_{p\in \Gamma}\delta_p$ on $\mathcal{G}$. In view of \eqref{eqsysj1} and \eqref{eqsysj2}, the vector field $\j_-$ satisfies
\begin{equation}\label{eqsysj}
\left\{
\begin{aligned}
&\di \j_-=2\pi (\nu-\rho)&&\mbox{in }\mathcal G\\
&\j_-\cdot\vnu=\varphi&&\mbox{on }\partial K_\l\\
&\j_-\cdot\vnu=0&&\mbox{on }\partial \mathcal G\backslash\partial K_\l
\end{aligned}
\right..
\end{equation}
and
\begin{equation}\label{eqnorm2j}
\int_{\mathcal{G}\backslash \cup B(p,\eta)}|\j_-|^2=\int_{\mathcal{G}\backslash \cup B(p,\eta)}|\j_1|^2+\vb \j_2\vb^2+2 \j_1\cdot \j_2.
\end{equation}
Using Lemma \ref{lemsumvf} combined with \eqref{eqenergyj1}, \eqref{eqnormj1}, \eqref{eqnormj2} and \eqref{eqnormj2bis}, we obtain
\begin{align*}
		W(E_-,\1_\mathcal G)\le&\, W(E_1,\1_\mathcal G)+\frac{1}{2}\norm E_2\norm^2_{L^2(\mathcal G)}+\norm E_1\norm_{L^{q'}(\mathcal G)}\norm E_2\norm_{L^{q}(\mathcal G)}\\
		\le&\, C(\l^{1+\alpha}+\l^{\alpha\left({2-\frac{2}{p}}\right)+\frac{2}{p}(2-\gamma)}+ \left(\norm \rho\norm_{C^{0,\lambda}(\mathcal G )}\l^{\alpha\lambda}\right)^2 \l^{1+3\alpha})\\
		&+ C \l^{\frac{1+\alpha}{q'}}{\left(\l^{\frac{\sigma}{q}}+\l^{\frac{1+\alpha}{q}+\alpha}\norm \rho\norm_{\mathcal C^{0,\lambda}}\l^{\alpha\lambda}\right)}\\
		\le&\,  C\left( \l^{1+\alpha}+\l^{\alpha\left({2-\frac{2}{p}}\right)+\frac{2}{p}(2-\gamma)}+\l^{\frac{1+\alpha}{q'}+\frac{\sigma}{q}}\right)\\
		&+C \left(\norm \rho\norm_{C^{0,\lambda}(\mathcal G )}\l^{\alpha\lambda}\right)^2 \l^{1+3\alpha}+C \left(\norm \rho\norm_{C^{0,\lambda}(\mathcal G )}\l^{\alpha\lambda}\right)\l^{1+2\alpha}\nonumber\\
\le&\, C\l^{1+\beta_1}+C\l^2 \left(\left(\norm \rho\norm_{C^{0,\lambda}(\mathcal G)}\l^{\alpha\lambda+{2\alpha-1}}\right)+\left(\norm \rho\norm_{C^{0,\lambda}(\mathcal G)}\l^{\alpha\lambda+\frac{3\alpha-1}{2}}\right)^2\right)\nonumber
\end{align*}
for some $\beta_1\in(0,1)$. Indeed, $$\alpha\left({2-\frac{2}{p}}\right)+\frac{2}{p}(2-\gamma)<2$$ and $$\frac{1+\alpha}{q'}+\frac{1}{q}\sigma<2\left(\frac{1}{q'}+\frac{1}{q}\right)=2.$$ {Moreover, if we choose $\alpha<\frac{2\lambda+1}{2\lambda+3}$, we obtain}
\begin{align}\label{eqenergyj}
W(E_-,\1_\mathcal G)
\le\, C\l^{1+\beta_1}+C\l^2 \left(\left(\norm \rho\norm_{C^{0,\lambda}(\mathcal G)}\l^{\beta_2\lambda}\right)+\left(\norm \rho\norm_{C^{0,\lambda}(\mathcal G)}\l^{\beta_3\lambda}\right)^2\right)
\end{align}
with $\beta_1,\beta_2,\beta_3\in (0,1)$. Such an $\alpha$ exists since $\frac{2\lambda+1}{2\lambda+3}\ge\frac{1}{2}>\frac{2-\gamma}{1+p}$ whenever $\lambda\ge\frac{1}{2}$. To summarize $\alpha$ must be chosen in the interval $\left(\frac{2-\gamma}{1+p},\delta\right)$ where $\delta=\min \left\{\frac{p-2+\gamma}{p-1}, \frac{2\lambda+1}{2\lambda+3} \right\}$. \smallskip

\noindent\textbf{Step 3.} There remains to define $t_-$ and  extend  $\j_-$ to $\mathcal D:= K_\l \backslash (\mathcal G\cup K_{t_-})$.
 We proceed as follows.
First of all, by a mean value argument we remark that if $\l$ is sufficiently large, there exists $t_-\in [\l-2\l^\alpha,\l-\l^\alpha]$ such that $\int_{\mathcal D} \rho(x)\,dx\in \NN$. 
 We then need to  partition $\mathcal D$ into rectangles $\mathcal D_i$ over which $\int_{\mathcal D_i} \ro $ is an integer. To do this, we repeat essentially the same as in Step 1, except we no longer have to deal with nonzero boundary conditions.
First, starting from the edges, we split 
  $\mathcal D$ into strips of width $\in [\hal\l^\a, \l^\a]$ on which $\int \ro$ is an integer, then we split again each strip into rectangles of sidelengths in $[\hal \l^\a, \l^\a]$ on which $\int \ro$ is an integer. This exhausts $\mathcal D$ since $\int_{\mathcal D} \ro\in \mn$, however because of the corners, some of the cells need to be ``L-shaped" polygons instead of rectangles (this doesn't cause any serious problem however).
  Once the $\mathcal D_i$ are thus constructed, we proceed to construct $\j_-$ in each $\mathcal D_i$, operating exactly as in $\mathcal R_i$ in  Step 1 and Step 2 (except replacing $\vp_i$ by $0$ everywhere) (i.e. splitting each $\mathcal D_i$ into many rectangles of size $O(1)$). We obtain  a vector field $\j_-$ in $\mathcal D$, with $\tilde \Gamma$ the corresponding set of points, satisfying an upper bound for $W(\j_-, \1_{\mathcal D})$ which is at least as good as \eqref{eqenergyj}.
  Combining this with the results of Steps 1 and 2, we obtain a vector field $\j_-$ satisfying the desired properties \eqref{defj-} and \eqref{energyj-}. We 
  also easily check that $\Lambda_-$ satisfies the other desired properties: its cardinal is bounded by the volume concerned i.e. $O(\l^{1+\alpha})$ and its points are separated by a fixed distance (depending only on $\overline \rho$), by \smallskip construction.

\noindent {\bf Step 4.}  There remains to prove \eqref{energyj-chi}. This follows the proof of \cite[Proposition 3.1]{sandierserfaty}.
First, we note that by definition of $W$ and  since the points in $\Lambda_-$ are separated by a fixed distance, say $4r_0$,
\begin{equation}\label{}
\hal \int_{\mathcal K^a_-\backslash \cup_{p \in \Lambda_-} B(p, r_0) } |\j_-|^2 \le W(\j_-, \1_{\mathcal K^a_-}) + \pi \# \Lambda_- \log \frac{1}{r_0} +C
\end{equation}
On the other hand 
\begin{multline}\label{mnb1}
W(\j_-, \chi)= \lim_{\eta\to 0}  \hal \int_{\mathcal K^a_-\backslash \cup_{p \in \Lambda_-} B(p, \eta) } |\j_-|^2 \chi+ \pi \sum_{p\in \Lambda_-} \chi(p) \log \eta
\\ 
\le  \hal  \int_{\mathcal K^a_-\backslash \cup_{p \in \Lambda_-} B(p, r_0) } |\j_-|^2+
\lim_{\eta\to 0} \sum_{p \in \Lambda_-} \(\hal \int_{B(p, r_0) \backslash B(p, \eta)} |\j_-|^2 \chi 
+  \pi  \chi(p) \log \eta\) \\
\le W(\j_-, \1_{\mathcal K^a_-}) + \pi \# \Lambda_- \log \frac{1}{r_0} +C
+ \lim_{\eta\to 0} \sum_{p \in \Lambda_-} \chi(p) 
\(\hal \int_{B(p, r_0) \backslash B(p, \eta)}|\j_-|^2 + \pi \log \eta\)\\ + \lim_{\eta\to 0} \sum_{p \in \Lambda_-} \hal\int_{B(p, r_0) \backslash B(p, \eta)} |\j_-|^2 (\chi - \chi(p)).\end{multline}
We may also bound 
$\sum_{p \in \Lambda_-} 
\( \int_{B(p, r_0) \backslash B(p, \eta)}\hal |\j_-|^2 + \pi \log \eta\)$ by $W(\j_-,\1_{\mathcal K^a_-})$.
It  remains to control the last term in the right-hand side of \eqref{mnb1}. Let us define 
$\Phi(t)= \hal \int_{\mathcal K^a_-\backslash \cup_{p \in \Lambda_-} B(p, t) } |\j_-|^2.$
We have $\Phi(t) \le W(\j_-,\1_{\mathcal K^a_-}) + \pi \# \Lambda_- \log \frac{1}{t}+C,$ and 
$\Phi'(t)= - \sum_{p \in \Lambda_-} \int_{\p B(p, t) }|\j_-|^2 $.
On the other hand, since $\chi$ is Lipschitz, we have 
\begin{multline*}
\sum_{p \in \Lambda_-} \int_{B(p, r_0) \backslash B(p, \eta)} |\j_-|^2 (\chi - \chi(p)) \le 
C \sum_{p \in \Lambda_-} \int_{B(p, r_0) \backslash B(p, \eta)} |\j_-|^2 |x-p|
\\ = - C \int_{\eta}^{r_0} \Phi'(t) t\, dt\\
= - C \( \Phi(r_0) r_0- \Phi(\eta) \eta+ \int_{\eta}^{r_0} \Phi(t)\, dt\) \le C W(\j_-,\1_{\mathcal K^a_-})+ C \# \Lambda_-+o_\eta(1),\end{multline*}
where $C$ depends only on $r_0$.
Inserting into \eqref{mnb1} we obtain 
$$W(\j_-, \chi)\le C W(\j_-, \1_{\mathcal K^a_-}) + C \# \Lambda_-$$
and combining with \eqref{energyj-} and \eqref{eqnumpointj-}, we obtain the result \eqref{energyj-chi}

\end{proof}

\subsection{Proof of Proposition \ref{existencemin}} Let $\{E_n\}_n$ be a minimizing sequence for $$\inf_{\mathcal B_{\rho,\varphi}(K_L(a))}W(E,\1_{K_L(a)}),$$ and let $\Lambda_n$ be the associated set of points, and $\nu_n= \sum_{p\in \Lambda_n} \delta_p$. 
In view of the constitutive relation \eqref{eqdefBrhoOmega}, the boundary data $\varphi$ and the density $\rho$ completely determine the number of points, so we have that $\{\nu_n\}_n$ is a  bounded  family of measures.
Up to extraction of a subsequence, we may write that $\nu_n$ converges weakly to some $\nu$ of the form $\sum_{p\in \Lambda} d_p \delta_p$ for some finite  set of points $\Lambda$ and some positive integers $d_p$.
\smallskip 

 { Next we show that, up to another extraction, $E_n$ converges in the sense of distributions to some $E$.
 Let $X$ be a smooth vector field vanishing  on $\partial K_L(a)$. We may write a Helmoltz decomposition $X=\nab \zeta+ \nab^\perp \xi$, with $\xi=0$ on $\partial K_L(a)$, and one may check that $\|\zeta\|_{L^\infty}$ is controlled by, say, the $C^1$ norm of $X$. Integrating by parts and using  the fact that $E_n\in B_{\rho,\varphi}(K_L(a))$, we have 
 \begin{multline*}
 \int_{K_L(a)} E_n \cdot X= \int_{K_L(a)} E_n \cdot (\nab \zeta+ \nab^\perp  \xi) = \int_{\partial K_L(a)} \varphi \zeta-2\pi \int_{K_L(a)} (\nu_n-\rho) \zeta  \\
 \le C \|\zeta\|_{L^\infty} \le C \|X\|_{C^1}.\end{multline*} 
 It thus follows that $E_n$ is bounded in the sense of distributions in $K_L(a)$ and converges, up to extraction, to some $E$, which must satisfy 
 $$\div E= 2\pi (\nu-\ro)\qquad \curl E=0.$$
 Letting $K_\eta= K_L(a) \backslash  \cup_{p\in \Lambda} B(p, \eta)$, we have that  $\div (E-E_n)= \curl (E-E_n)= 0$ in $K_\eta$ for $n$ large enough depending on $\eta$ (such that  $\Lambda_n \cap K_\eta= \emptyset$).
 Elliptic regularity then implies that 
 the convergence of $\j_n$  is in fact uniform in $K_{2\eta}$, i.e.  away from $\Lambda$.}
\smallskip 

We next claim that $\Lambda \cap \p K_L(a)=\emptyset$. If not there would be some point $p_n\in \Lambda_n$ such that $\dist(p_n, \p K_L(a))\to 0$. Up to a change  of coordinates, we may assume that $p_n \to 0\in \p K_L(a)$ and that the segment $[0, \ep_0]\times \{0\}$ is included in $\p K_L(a)$ and contains no other point of $\Lambda$. 
An expansion of $\j_n$ near $p_n$ shows that 
\begin{equation}\label{vpbord1}
\j_n(x)= \frac{x-p_n}{|x-p_n|^2} + f_n(x)\end{equation} hence
\begin{equation}\label{vpbord}
\vp = \j_n \cdot \vnu= \frac{(x-p_n) \cdot \vnu}{|x-p_n|^2} + f_n(x)\cdot \vnu \quad \text{on } [0,\ep_0]\times \{0\},\end{equation}
where $f_n$ converges uniformly on $[0,\ep_0]\times \{0\}$.
On the other hand, computing explicitly, we have, for any $\ep<\ep_0$, 
$$\int_{[0,\ep]\times \{0\}} \left| \frac{(x-p_n) \cdot \vnu}{|x-p_n|^2}\right|\, dx= 
\int_0^\ep\frac{p_n^{(2)}}{t^2 + p_n^{(2)}} \, dt= \mathrm{Arctg} \, \frac{\ep}{p_n^{(2)}}$$
where $p_n^{(2)}$ is the second coordinate of $p_n$ and tends to $0$.
Since $\vp\in L^p(\p K_L(a))$ and $f_n$ converges uniformly, inserting into \eqref{vpbord} we must have 
$\mathrm{Arctg} \, \frac{\ep}{p_n^{(2)}}\le Co_\ep(1)$ as $\ep \to 0$,  where $C$ is independent of $n$ and $\ep$. But for any fixed $\ep$, $\mathrm{Arctg} \, \frac{\ep}{p_n^{(2)}}\to \frac{\pi}{2}$ as $n \to \infty$, a contradiction. It follows that the claim holds.
From this we deduce two facts: first    $\j_n$ converges uniformly in a neighborhood of $\p K_L(a)$, therefore $\j$ satisfies $\j\cdot \vnu = \vp $ on $\p K_L(a)$ as well.  { Second, since $\Lambda\cap K_L(a) = \emptyset$ and $\Lambda$ is a finite set of points,  the points in  $\Lambda_n$ stay at distance from $\p K_L(a)$ which is bounded below,  hence  the argument in the proof of  \cite[Lemma 4.8]{sandierserfaty} allows to show that $d_p=1$ for every $p\in \Lambda$, i.e. there are no multiple points. We may thus conclude that $\j$   belongs to $\mathcal B_{\rho,\varphi}(K_L(a))$. }
\smallskip

 To complete the proof there  remains to show that
\begin{equation*}
	\lim_{n\to+\infty} W(\j_n,\1_{K_L(a)})=W(\j,\1_{K_L(a)}).
\end{equation*}
But using again  that the points in  $\Lambda_n$ stay at distance from $\p K_L(a)$ which is bounded below, this follows easily by the argument of \cite{sandierserfaty}, end of the proof of Lemma 4.8: since $\j_n \to \j$ uniformly away from $\Lambda$, we have  for any $\eta>0$,
$$\lim_{n\to \infty} \frac12\int_{K_L(a)\sm \cup_{p\in \Lambda_n} B(p, \eta)} |\j_n|^2 +\pi  \# \Lambda_n \log \eta = \frac12\int_{K_L(a)\sm \cup_{p\in \Lambda} B(p, \eta)} |\j|^2 +\pi  \# \Lambda \log \eta$$
and the result will follow if the convergence is uniform with respect to $\eta$ (so that we may reverse the $\eta\to 0$ and $n\to \infty$ limits).
This follows, as in \cite{sandierserfaty} from the expansion \eqref{vpbord1} with $f_n$ uniformly bounded near each $p_n$, implying
$$\left|\hal \int_{B(p_n, \eta_2)\sm B(p_n,\eta_1)}\vb \j_n\vb^2+ \pi \log \frac{\eta_1}{\eta_2}\right|\le C \sqrt{\eta_2}$$
thus proving the uniformity of the convergence.

We note that the existence of a minimizer of $W(\j, \mathbb{T}_L)$ on $\mathcal A_{m,per}(\mathbb{T}_L)$ can be proven with the same arguments.


\subsection{Proof of Proposition \ref{proplimmincon}}
With the help of Proposition \ref{propj+j-}, we have a first easy comparison between the quantities $\sigma_0(\cdot;\rho)$ and $\sigma_\varphi(\cdot;\rho)$.

\begin{lem}\label{corestimationmin} Under the hypotheses of Proposition \ref{propj+j-}, there exist $c,C$ positive constants and $\beta\in(0,1)$ such that for all $\l\ge c$, there exist
$t_+\in[\l+\l^\a,l+2l^\a]$ and $t_-\in[\l-2\l^\a,\l-\l^\a]$ such that
\begin{align}\label{eqestimationminlow}
\sigma_0(K_{t_+}(a);\rho)-C  \left(\frac{1}{\l^{1-\beta}}+\left(\norm \rho\norm_{C^{0,\lambda}(\mathcal K^a_+)}\l^{\beta\lambda}\right)\left(1+\norm \rho\norm_{C^{0,\lambda}(\mathcal K^a_+)}\l^{\beta\lambda}\right)\right)\le \sigma_\varphi(K_\l(a);\rho),\\
\label{eqestimationminup}
\sigma_\varphi(K_\l(a);\rho)\le \sigma_0(K_{t_-}(a);\rho)+ C  \left(\frac{1}{\l^{1-\beta}}+\left(\norm \rho\norm_{C^{0,\lambda}(\mathcal K^a_-)}\l^{\beta\lambda}\right)\left(1+\norm \rho\norm_{C^{0,\lambda}(\mathcal K^a_-)}\l^{\beta\lambda}\right)\right),
\end{align}
where $\mathcal K_+^a=K_{t_+}(a)\backslash K_\l(a)$ and $\mathcal K_-^a=K_{\l}(a)\backslash K_{t_-}(a)$.
The  constants $c,C$ and $\beta<1$ depend only on $p$, $\gamma$, $M$, $\underline \rho$ and $\overline \rho$.
\end{lem}

\begin{rem} Note that if $\rho$ is constant, Lemma \ref{corestimationmin} means that
\begin{equation}\label{eqestiminconst}
\sigma_0(K_{t_+}(a);\rho)-C{\l^{\beta-1}}\le  \sigma_\varphi(K_\l(a);\rho)\le \sigma_0(K_{t_-}(a);\rho)+ C {\l^{\beta-1}}
\end{equation}
\end{rem}


\begin{proof}[Proof of Lemma \ref{corestimationmin}]  Without loss of generality, we may assume $a=0$.

Let $\j$ be  a minimizer for $\sigma_\varphi(K_\l;\rho)$, i.e. $W(\j,\1_{K_\l})=\sigma_\varphi(K_\l;\rho)\vb K_l\vb$. Let us extend $E$ to $K_{t_+}\backslash K_\l$ via Proposition \ref{propj+j-}: we set $E=E_+$ in $K_{t_+}\backslash K_\l$ where $E_+$ is given by the first part of Proposition \ref{propj+j-} and thus satisfying \eqref{energyj+}. This extended $\j$ satisfies the hypothesis of Lemma \ref{lemdiv0} on $K_{t_+}$ with $\varphi=0$ (indeed note again that the normal components coincide on $\partial K_\l$); hence, there exists $\j\in \mathcal B_{\rho,0}(K_{t_+})$ such that $W(\tj,\1_{K_{t_+}})\le W(\j,\1_{K_{t_+}})$. As a consequence, using \eqref{energyj+},
\begin{align*}
\sigma_0(K_{t_+};\rho)&\vb K_{t_+}\vb \le W(\tj,\1_{K_{t_+}})\le W(\j,\1_{K_{t_+}})=W(\j,\1_{K_{\l}})+W(\j,\1_{K_{t_+}\backslash K_\l})\\
&\le \sigma_\varphi(K_{\l};\rho)\vb K_{\l}\vb + C \left(\l^{1+\beta}+\l^2\left(\norm \rho\norm_{C^{0,\lambda}(\mathcal K_+)}\l^{\beta\lambda}\right)\left(1+\norm \rho\norm_{C^{0,\lambda}(\mathcal K_+)}\l^{\beta\lambda}\right)\right).
\end{align*}

Next, let $\j$ be a minimizer of $\sigma_0(K_{t_-};\rho)$ in $K_{t_-}$ and $\j=\j_-$ in $K_\l\backslash K_{t_-}$ with $\j_-$ constructed via the second part of Proposition \ref{propj+j-}. This extended $\j$ satisfies the assumptions of Lemma \ref{lemdiv0} as before, thus there exists $\tj\in \mathcal B_{\rho,\varphi}(K_\l)$ such that $W(\tj,\1_{K_{\l}})\le W(\j,\1_{K_{\l}})$. Consequently, using \eqref{energyj-}, we have
\begin{align*}
\sigma_\varphi(K_{\l};&\rho)\vb K_{\l}\vb \le W(\tj,\1_{K_{\l}})\le W(\j,\1_{K_{\l}})=W(\j,\1_{K_{t_-}})+W(\j,\1_{K_{\l}\backslash K_{t_-}})\\
&\le \sigma_0(K_{t_-};\rho)\vb K_{t_-}\vb +C \left(\l^{1+\beta}+\l^2\left(\norm \rho\norm_{C^{0,\lambda}(\mathcal K_-)}\l^{\beta\lambda}\right)\left(1+\norm \rho\norm_{C^{0,\lambda}(\mathcal K_-)}\l^{\beta\lambda}\right)\right).
\end{align*}
\end{proof}

We now have  all the ingredients to prove the following proposition that gives a comparison between the quantities $\sigma_0(\cdot;m)$, $\sigma_\varphi(\cdot;m)$, $\sigma_{per}(\cdot,m)$ and $\sigma^*_m=\min_{\mathcal A_m} W$ for  $m$ a positive number. This correspond to the case of a constant background; the case of a non-constant background, useful for the Coulomb gas problem,  is more involved and  is treated in Section \ref{sectionnonconstant}.



\begin{proof}[Proof of Proposition \ref{proplimmincon}] Without loss of generality, we may assume $a=0$ and $m=1$; the case of general $m$ follows by \eqref{eqscalingW}. We reproduce here arguments of \cite[Section 4]{sandierserfaty}.

First, let us show that $\sigma_1^*\le \sigma_0(K_L;1)$. Let $\j_0$ be a minimizer of $ \sigma_0(K_L;1)$.
By Remark \ref{reminterseclambda},  $\Lambda\cap \partial K_L=\emptyset$. Moreover, since $\curl\,\j_0=0$ in $K_L$ and $\j_0\cdot\vnu=0$, we may write $\j_0=-\nabla H_0$ with $\partial_{\vnu} H_0=0$ on $\partial K_L$. Thus defining $H_0$ on $[-L,3L]\times[-L,3L]$ by reflections  across the sides of the square $K_L$ we have $-\Delta H_0=2\pi\left(\sum_{p\in \Lambda_0} \delta_p-1\right)$, where $\Lambda_0$ is obtained from $\di\,\j_0+2\pi$ by reflections  across the sides of the square $K_L$. Moreover, $H_0(-L,y)=H_0(3L,y)$ and $H_0(x,-L)=H_0(x,3L)$ so that we may periodize $H_0$ to have it defined on $\RR^2$.
Then $\j:=-\nabla H_0\in \mathcal A_{1}$ and since everything is periodic, $W$ can be computed through the results of \cite[Section 3.1]{sandierserfaty}:
$$
\sigma_1^*\le W_K(\j)=\frac{W(\j,\1_{K_{2L}})}{4 \vb K_L\vb}=\frac{4 W(\j_0,\1_{K_L})}{4\vb K_L\vb}=\sigma_0(K_L;1).
$$
Next, let $\j$ be a minimizer  of $W$ in $\mathcal A_1$. As a consequence of \cite[Proposition 4.2]{sandierserfaty}, there exists a sequence $\{\j_L\}_{L^2\in \NN}\subset\mathcal A_{1,0}(K_L)$ such that
\begin{equation*}
\frac{W(\j_L,\1_{K_L})}{\vb K_L\vb}\le \sigma_1^*+o(1)_{L\to +\infty}.
\end{equation*}
 By definition of $\sigma_{0}(K_L;1)$, it follows that
$
\sigma_0(K_L;1)\le \sigma_1^*+o(1)_{L\to +\infty}.
$

 Similarly, as a consequence of \cite[Corollary 4.4]{sandierserfaty}, there exists a minimizing sequence for $\min_{\mathcal A_1}W$ consisting of $\mathbb{T}_L$-periodic vector fields. Hence, since $\sigma_{per}(L;1)$ is the minimum among all $\j$ which are $\mathbb{T}_L$ periodic, we deduce
$
\sigma_{per}(L;1)\le \sigma_1^*+o(1)_{L\to +\infty}.
$
Moreover, it is clear that $\sigma_1^*\le \sigma_{per}(L;1)$. Indeed, $\sigma_{per}(L;1)=\min\limits_{\j\in \mathcal A_{1,per}(\mathbb T_L)}W(\j)\ge \min_{\mathcal A_1}W=\sigma_1^*$.

Finally, to obtain (\ref{eqlimsigmaphicon}) with $m=1$, we use (\ref{eqlimsigma0con}) with $m=1$ and combine it with Lemma \ref{corestimationmin}.
\end{proof}

\section{Proof of Theorem \ref{thmmaincb} and Theorem \ref{thmmainper}}
\label{sec3}

In this section, we turn to the main proofs. We use previous results in  the case where the function $\rho$ is constant, that means $\rho(x)=m$ for all $x\in\RR^2$ and for some positive number $m$. Note that in this particular case the condition \eqref{h1boundary} reduces to
\begin{equation}
\label{h1boundary}
\frac{1}{2\pi}\int_{\partial K_L(a)}\varphi+ m\vb K_L(a)\vb \in \NN.
\tag{$\mathrm{HB}_1$}
\end{equation}

\subsection{Preliminary results}

We start by recalling some technical results from \cite{sandierserfaty} that we will need repeatedly. 

The following result from  \cite[Proposition 4.9]{sandierserfaty} is crucial because it shows how the energy density associated to $W(\j, \chi)$ can be replaced by an essentially positive density, at a cost depending on the number of points near the boundary (this is why we always need to have good controls on the number of points in addition to control on $W(\j, \chi)$).
\begin{prop}[\cite{sandierserfaty}]\label{prop49ss}
Assume $U\subset \mathbb R^2$ is open, and let $\hat U$ denotes the set $\{x\vb d(x,U)<1\}$. Assume  $(\j,\nu)$ are such that $\nu=\sum_{p\in \Lambda}\delta_p$ for some finite subset $\Lambda$ of $\widehat U$ and $\di \j=2\pi(\nu-\ro)$ in $\widehat U$ and { $\curl E =  0 $ in $\cup_{p\in \Lambda} B(p,\eta)\cap \widehat{U}$ for some $\eta>0$, with $\ro\in L^{\infty}(\widehat U)$. } Then there exists a measure $g$ supported on $\widehat U$ and such that
\begin{itemize}
\item $g\ge -C(\norm \ro\norm_{\infty}^2+1)$ on $\widehat U$, where $C$ is a universal constant.
\item For any function $\chi$ compactly supported in $U$ we have
\begin{equation}\label{eqestimWg}
\left\vb W(\j,\chi)-\int \chi\,dg\right\vb\le Cn(\log n+ \norm \ro \norm_{\infty})\norm\nabla\chi\norm_{\infty}
\end{equation}
where $n=\#\{p\in\Lambda\vb B(p,1)\cap\mathrm{supp}(\nabla\chi)\neq\emptyset\}$.
\item For any $S\subset U$
\begin{equation}\label{eqnumberpoint}
\#(\Lambda \cap S)\le C(1+\norm \ro\norm_{\infty}^2\vb \widehat S\vb +g(\widehat S)),
\end{equation}
where $C$ is universal.
\end{itemize}
\end{prop}
We note here that in the proof of \cite{sandierserfaty} we can control the error in \eqref{eqestimWg} by the number of points at distance $\le 1$ from $\supp (\nab \chi)$, at the expense of a possibly larger constant $C$ in the first item.
{ We also note that in \cite[Proposition 4.9]{sandierserfaty} the additional assumption $\curl E=0$ is placed, however it is not used at all in the proof, so we removed it here.}

The next lemma gives a control of $L^p$ norm of vector field by the renormalized energy $W$. A better estimate can be found in \cite{st,rs} but we will not need them here.

\begin{lem}[{\cite[Lemma 4.7]{sandierserfaty}}]
	\label{lemnormp} Let $\chi$ be a positive function compactly supported in an open set $U$ and assume that $\di \j=2\pi(\nu-\rho)$ in
	\begin{equation*}
		\hat U=\{x\vb d(x,U)<1\}
	\end{equation*}
	where $\nu=\sum_{p\in \Lambda}\delta_p$ for some finite subset $\Lambda$ of $\hat U$, and { $\curl E =  0 $ in $\cup_{p\in \Lambda} B(p,\eta)\cap \widehat{U}$}. Then, there exists $C>0$ universal and for any $p\in[1,2)$, $C_p>0$ depending only on $p$, such that
	\begin{equation*}
		\int_U\chi^{p/2}\vb \j\vb^{p}\le C(\vb U\vb +C_p)^{1-p/2}(W(\j,\chi)+n(\log n +\norm \rho\norm_{\infty})\norm \chi\norm_{\infty}+n\norm \nabla \chi\norm_{\infty})^{p/2}
	\end{equation*}
	where $n=\# \Lambda$.
\end{lem}
From these two results we can deduce the following  preliminary result which allows to bound the number of points in an ``annular type" region.
\begin{lem}
\label{lembor} Let $\ro \in L^\infty(K_{L+1})$ and { $\j$ satisfy 
$$\div \j = 2\pi (\nu- \rho) \quad \text{in}  \ K_{L+1}$$
with $\nu =\sum_{p \in \Lambda} \delta_p$ for some discrete set $\Lambda$,  and $\curl \j=0$ in $\cup_{p\in \Lambda} B(p, \eta)$ for some $\eta>0$,}
and assume $$W(\j, \chi_{K_L}) \le C_0 L^2 \qquad \nu(K_{L+1}) \le C_0 L^2.$$
 Then  for any positive $r, c_1, c_2$ such that $r+ c_1 \le L-2$, and $p\in (1,2)$,  we have
$$\left|\nu(K_{r+c_1})-\nu(K_{r-c_2})\right|\le C
r^{1-\frac{1}{p}} L^{\frac{2}{p}}\log^{\frac{1}{2}}L+  Cr
$$
where $C$ depends only on $C_0, c_1, c_2, \|\ro\|_{L^\infty}$  and $p$.
\end{lem}
\begin{proof}
First of all, Lemma \ref{lemnormp} gives that for $p\in (1,2)$
\begin{align}\label{eqestimnormp}
\norm \j\norm_{L^p(K_{L-1})} \le&  CL^{\frac{2}{p} -1}(W(\j,\chi_{K_L})+\nu(K_{L+1})\log \nu(K_{L+1}))^{1/2}\nonumber\\
\le&
C L^{\frac{2}{p}-1}(L^2+L^2\log L)^{\frac{1}{2}}\le  C L^{\frac{2}{p}}\log^{\frac{1}{2}}L,
\end{align}
where $C$ depends only on $p$, $\norm \rho\norm_{L^\infty}$, and $C_0$.

 Next, for $t\le L-2$, a mean value argument gives a $t_-\in [t-1,t]$ and a $t_+\in [t,t+1]$ such that
 \begin{equation}\label{jp}
 \|\j\|_{L^p(\p K_{t_-})} \le \|\j\|_{L^p(K_{L-1})}\le  C L^{2/p}\log^{\frac{1}{2}}L
\qquad \|\j\|_{L^p(\p K_{t_{+}})} \le  C L^{2/p}\log^{\frac{1}{2}}L.\end{equation}
Since $\nu(K_{t_{\pm}}) - \int_{K_{t_\pm}} \ro(x)\, dx= \frac{1}{2\pi}\int_{\p K_{t_\pm}} \j \cdot \vnu$, it follows with H\"older's inequality that
$$\left|\nu(K_{t_{\pm}}) - \int_{K_{t_\pm}} \ro(x)\, dx\right|\le C t^{1-\frac{1}{p}} L^{2/p}\log^{\frac{1}{2}}L.
$$
Since $\nu(K_{t_-})\le \nu(K_t)\le \nu(K_{t_+}) $ and  $\ro \in L^\infty$ we immediately deduce that  for any $t\le L-2$
\begin{equation}
\label{nkt}
\left|\nu(K_{t}) - \int_{K_{t}} \ro(x)\, dx\right|\le C t^{1-\frac{1}{p}} L^{2/p}\log^{\frac{1}{2}}L+ \|\ro\|_{L^\infty} t.
\end{equation}
The result easily follows  with again the  same argument.

\end{proof}

The proof of Theorem \ref{thmmaincb} relies on the  selection of  good boundaries by mean value arguments, which is provided by the following two variants of the same lemma, whose proofs are postponed to the end of the section. The first one concerns vector-fields without boundary conditions, the second one vector fields with given good boundary conditions. Because we will need them later, we state them with varying background.

\begin{lem}\label{lemgoodboundary} Let $p\in (1,2)$, $\gamma\in\left(\frac{3-p}{2},1\right)$, and $K_L$ be some square of sidelength $2L$. Let $\rho\in L^{\infty}(K_{L+1})$ and $\j$ { satisfy 
$$\div \j = 2\pi (\nu- \rho) \quad \text{in}  \ K_{L+1}$$
with $\nu =\sum_{p \in \Lambda} \delta_p$ for some discrete set $\Lambda$, and $\curl \j=0$ in $\cup_{p\in \Lambda} B(p, \eta)$ for some $\eta>0$; and assume that}
 there exist $C_1,C_2>0$ such that  we have
\begin{equation}\label{eqcondnusquare}
 \frac{\nu(K_{L+1})}{\vb K_{L+1}\vb}<C_1
\end{equation}
for the associated $\nu$ and such that
\begin{equation}\label{eqcondenergy}
\frac{W(\j,\chi_{K_L})}{\vb K_L\vb}\le C_2.
\end{equation}
Then, for all $L$ large enough (depending on $\gamma$) and any $\l$ such that
\begin{equation}\label{eqconddelta}
L\ge \l\ge L^{1/\delta} \mbox{ with } 1<\delta < \frac{1}{p}\left(1-\gamma+\frac{p}{2}(1+\gamma)\right)
\end{equation}
and $a\in K_L$ such that $K_\l(a)\subset K_L$, there exists some $t\in[\l-2\l^\gamma,\l-\l^{\gamma}]$ such that
\begin{equation}\label{eqgoodboundary}
\int_{\partial K_{t}(a)} \vb \j\vb^p \le M \l^{2-\gamma},
\end{equation}
\begin{equation}\label{eqgoodboundarypropenergy}
W(\j,\1_{K_{t}(a)}) \le W(\j,\chi_{K_\l(a)})+ C\l^{1+\gamma}+C\l^{\frac{2\delta -1}{p}+1}\log^{\frac{3}{2}}\l,
\end{equation}
\begin{equation}\label{eqgoodboundaryproppoints}
\left\vb\nu(K_{t}(a))-\int_{K_{t}(a)}\rho(x)\,dx\right\vb\le  C \l^{2-\gamma}.
\end{equation}
The constants $C, M$ depend only $p$, $\gamma, \delta$,  $C_1$, $C_2$ and $\norm \rho\norm_{L^\infty(K_{L+1})}$.
\end{lem}

\begin{lem}
	\label{lemcenergygb} Let $p\in(1,2)$, $\gamma\in\left(\frac{3-p}{2},1\right)$, $M>0$, $L>0$ and $K_L$ be some square of sidelength $2L$. Let $\rho\in L^{\infty}(K_L)$ and $\varphi\in L^p(\partial K_L(a))$ such that \eqref{h1boundary} and \eqref{h2boundary} are satisfied in $K_L$. { Let $\j$
	satisfy 
$$\begin{cases} \div \j = 2\pi (\nu- \rho) \quad &  \text{in}  \ K_{L}\\
\j \cdot \vnu=\varphi & \text{on} \ \p K_L\end{cases}$$
with $\nu =\sum_{p \in \Lambda} \delta_p$ for some discrete set $\Lambda$, and $\curl \j=0$ in $\cup_{p\in \Lambda} B(p, \eta)$ for some $\eta>0$; and assume that}
	 there exist $C_1,C_2>0$ such that
	\begin{equation}\label{eqcondnusquareext}
		\frac{\nu(K_{L})}{\vb K_{L}\vb}\le C_1
	\end{equation}
	for the associated $\nu$ and such that
	\begin{equation}\label{eqcondenergy1}
		\frac{W(\j,\1_{K_L})}{\vb K_L\vb}\le C_2.
	\end{equation}
	Then, for all $L$ large enough,  we have
	\begin{equation}
		\label{eqcontrolenergychi1}
		W(\j,\chi_{K_L})\le W(\j,\1_{K_L})+\tilde C L^{1+\beta},
	\end{equation}
and 	for all $a\in K_L$ such that $K_\l(a)\subset K_L$
	\begin{equation}
		\label{eqcontrolenergy1}
		W(\j,\chi_{K_l(a)})\le W(\j,\1_{K_L})+CL(L-\l)+\tilde C L^{1+\beta}
	\end{equation}
	where $C$ is a universal constant, $\beta\in(0,1)$ and $\tilde C$ depends on $p$, $M$, $\gamma$, $C_1$ and $\norm \rho\norm_{L^\infty}$.
	Moreover, the results of Lemma \ref{lemgoodboundary}, i.e.  for any $\l$ satisfying \eqref{eqconddelta}, the existence of $t$ with  \eqref{eqgoodboundary}, 
	\eqref{eqgoodboundarypropenergy}, 
	\eqref{eqgoodboundaryproppoints},   hold true.
\end{lem}

\subsection{Proof of Theorem \ref{thmmaincb}} First of all, we observe that \eqref{eqlimsigma0cb} and \eqref{eqlimsigmaphicb} have been  proven in Proposition \ref{proplimmincon}. There remains to show the equidistribution properties of a minimizer $\j_\vp$.
The proof is based on a bootstrap argument:
by a mean value argument, using the a priori bound on the energy $W(\j_\vp,\indic_{K_L(a)}) \le C L^2$ and the fact that $W$ essentially controls the $L^p$ norm of $\j_\vp$ for $p<2$ (see Lemma \ref{lemnormp}), we can find a square close to $K_\l(a)$  which has a good boundary, i.e. such that \eqref{h2boundary} is satisfied (relative to $\l$). This is only possible if $\l$ is not too small compared to $L$, more precisely if $\l \ge L^{1/\delta}$ for some $\delta>1$ specified later. If indeed $\l\ge L^{1/\delta}$ then we are essentially done: a simple comparison argument in the square with the good boundary, combined with \eqref{eqlimsigmaphicon} allows to conclude that the energy (per unit volume) of $\j_\vp$ in the square is close to $\sigma_m^*$, and the number of points per unit volume is close to $m$.
If $\ell$ is smaller than $L^{1/\delta}$, then we bootstrap the argument: we first obtain by the above argument a control of the energy and the number of points on a square of size $L^{1/\delta}$ containing $K_\ell(a)$, and then we re-apply the reasoning starting from that square. This allows to go down to $\ell \le L^{1/\delta^2}$, and we iterate the procedure until we reach any $\ell$. This iteration will not cumulate any error, its only main restriction is that the final square will have to be at a certain distance away from $\p K_L(a)$, because of the repeated mean value arguments.
This restriction is natural however, since a boundary condition $\vp $ satisfying \eqref{h2boundary} can concentrate locally on $\p K_L(a)$, and it then takes a certain distance for a minimizer to ``absorb" the effect of such a concentration on the boundary.


We start with the following easy comparison  lemma.
\begin{lem}\label{propestimationgb}
Let $m$ be a positive number. Let $q\in (1,2)$, $L>0$ and  $\varphi\in L^q(\partial K_L)$ such that \eqref{h1boundary} is satisfied in $K_L$. Let
 $\j_\varphi$ be a minimizer of $W(\j,\1_{K_L})$ in  the class $\mathcal{A}_{m,\varphi}(K_L)$.
For every $a\in K_L$ such that $K_\l(a)\subset K_L$, let $\psi$ be the restriction of $\j_\varphi\cdot \vnu$ to $\partial K_\l(a)$.
If $\psi$ satisfies \eqref{h2boundary} in $K_\l(a)$ for some $p\in(1,2)$, $\gamma\in\left(\frac{3-p}{2},1\right)$ and $M>0$, then
\begin{equation}
\left\vb\frac{W(\j_\varphi,\1_{K_{\l}(a)} )}{\vb K_{\l(a)}\vb}-\sigma_m^*\right\vb\le o(1)_{\l\to+\infty}.
\end{equation}
\end{lem}

\begin{proof} Since $\j_\varphi$ is a minimizer, it must be a minimizer on $K_\l(a)$ with respect to its own boundary condition, i.e. we have
$$
 \sigma_\psi (K_\l(a);m) = \frac{W(\j_\varphi,\1_{K_{\l}(a)} )}{\vb K_{\l(a)}\vb}.
$$
The result then follows  by applying Proposition \ref{proplimmincon}.
\end{proof}

\begin{rem}\label{aprelem31}
We note that the results of Proposition \ref{proplimmincon}
and Lemma \ref{propestimationgb} still hold true if $K_L(a)$ and   $K_\l(a)$  are not squares but rectangles whose sides are both  comparable to $L$ as $L\to\infty$.\end{rem}

We now proceed to the proof of Theorem \ref{thmmaincb}.

\noindent
{\bf Step 1.} Proof of \eqref{eqenergycb} in the general case.
\\
Let $p\in(1,2)$, $\gamma\in\left(\frac{3-p}{2},1\right)$ and $\delta>0$  as in \eqref{eqconddelta}; let $\l\le L$ and $a\in K_L$ such that $K_\l(a)\subset K_L$.  
Let $\j_\vp$ be a minimizer for $\sigma_\vp(K_L; m)$.  
In view of \eqref{eqlimsigmaphicon}, we have $W(\j_\vp, \indic_{K_L}) \le (\sigma_m^* +1) \vb K_L \vb$ for $L$ large enough.  Moreover, the number of points (associated to $E_\vp$) in $K_L$ is completely determined by $\vp$, it is $\frac{1}{2\pi} \int_{\p K_L} \vp +m|K_L|$. In view of \eqref{h2boundary}, it is thus controlled by $(m+1) \vb K_L\vb$, for $L$ large enough.
{ Let us assume first that $d(K_\l(a), \p K_L) \ge 3\l^\gamma$ so that 
$d(K_\l(a), \p K_L) \ge \max(3\l^\gamma, L^\beta)$.
\\

{\it Case 1:} $\l + 3\l^\gamma \ge L^{1/\delta}$.  Let us then define the smaller scale $\l_1= \l +  3\l^\gamma$. Since we assumed $d(K_\l(a), \p K_L) \ge 3\l^\gamma$, we have $K_{\l+3\l^\gamma}(a)\subset K_L$, and so 
 there exists a center $a_1$ such that 
 \begin{equation}\label{pincsq}
K_{\l+3\l^\gamma}(a)\subset K_{\l_1}(a_1)\subset K_L
. \end{equation}
If $L$ is large enough (hence $\l$ is too), we may then apply Lemma \ref{lemcenergygb} to $\j_\varphi$ in $K_L$ with smaller square $K_{\l_1}(a_1)  \subset K_L$, $C_1=m+1$ and $C_2=\sigma^*_m+1$.
The lemma gives us the  existence of  $t_1$ satisfying $\l_1-2\l_1^{\gamma}\le t_1<\l_1-\l_1^\gamma$ and such that 
$$
\int_{\partial K_{t_1}(a_1)} \vb \j_\varphi \vb^p \le M_1 \l_1^{2-\gamma}
$$
with $M_1$ and $C$  depending only on $C_1$, $C_2$ and $p$. Moreover,
\begin{align}\label{nukt1}
\left\vb\frac{\nu(K_{t_1}(a_1))}{\vb K_{t_1}(a_1)\vb}-m\right\vb\le \frac{C} {\l_1^{\gamma}}
\end{align}
We note that $t_1\ge \l_1-2\l_1^{\gamma}= (\l+3\l^\gamma) -2 (\l+3\l^\gamma)^\gamma\ge \l$ as soon as $\l$ is large enough, thus, with \eqref{pincsq} we have 
\begin{equation}\label{pincsq2} K_\l(a)\subset K_{t_1}(a_1)\subset K_{\l_1}(a_1) \subset K_L.\end{equation}
Applying then Lemma  \ref{propestimationgb} in the square $K_{t_1}( a_1)$, we deduce
\begin{equation}\label{wt1}
\left\vb\frac{W(\j_\varphi,\1_{K_{t_1}(a_1)} )}{\vb K_{t_1}( a_1)\vb}-\sigma_m^*\right\vb\le o(1)_{\l_1\to+\infty}.
\end{equation}
Applying then Lemma \ref{lemcenergygb} in this same square, 
more precisely applying \eqref{eqcontrolenergy1}  and combining with \eqref{wt1}, we obtain  \eqref{eqenergycb}. 
\\

{\it Case 2:}  $\l + 3\l^\gamma \le L^{1/\delta}$.  Let 
 $\l_1= L^{1/\delta}$.
Since we have  $d(K_\l(a), \p K_L) \ge L^\beta$, we have $K_{\l+L^\beta}(a)\subset K_L$, and so 
 there exists a center $a_1$ such that
 \begin{equation}\label{pincsqp}
K_{\l+L^\beta}(a)\subset K_{\l_1}(a_1)\subset K_L
. \end{equation}
If $L$ is large enough, we then apply Lemma \ref{lemcenergygb} to $\j_\varphi$ in $K_L$ with smaller square $K_{\l_1}(a_1)  \subset K_L$, $C_1=m+1$ and $C_2=\sigma^*_m+1$.
The lemma gives us the  existence of  $t_1$ satisfying $\l_1-2\l_1^{\gamma}\le t_1<\l_1-\l_1^\gamma$ and such that 
$$
\int_{\partial K_{t_1}(a_1)} \vb \j_\varphi \vb^p \le M_1 \l_1^{2-\gamma}
$$
with $M_1$ and $C$  depending only on $C_1$, $C_2$ and $p$, and 
\eqref{nukt1} holds.   
Applying  Lemma  \ref{propestimationgb} in  $K_{t_1}( a_1)$, we have \eqref{wt1} in that square.
We note that we have $t_1-\l_1  \ge  -2\l_1^\gamma =  - 2L^{\gamma/\delta}     $  so in view of \eqref{pincsqp} we have 
$$K_{ \l+L^\beta- 2 L^{\gamma/\delta}}    (a)\subset K_{t_1}(a_1)\subset K_{\l_1}(a_1) \subset K_L.$$
Next, 
we  observe  that if $\l \ge \bar c$ for some $\bar c>0$, the assumptions of Lemma \ref{lemcenergygb} are satisfied in $K_{t_1}(a_1)$ with the same constants $C_1$ and $C_2$.  This is an immediate consequence of \eqref{nukt1} and \eqref{wt1}, as soon as $t_1$, hence $\l$ is large enough. 
We can thus re-apply Lemma \ref{lemcenergygb}  in $K_{t_1}(a_1)$ with new subscale 
$\l_2 = \max(\l+ 3\l^\gamma, \l_1^{1/\delta})$. We may distinguish the two cases $\l+ 3\l^\gamma\ge \l_1^{1/\delta}$ and $\l+ 3\l^\gamma\le \l_1^{1/\delta}$ just as above, and iterate the same proof. This way, we define a finite sequence $\l_k$ with $\l_k=\max(\l+ 3\l^\gamma, \l_{k-1}^{1/\delta}$) with terminates at $ \l + 2\l^\gamma$. This takes at most $s$ steps where \begin{equation}
\label{defs}
s := - \left[\frac{\log \frac{\log \l}{\log L} } {\log \delta}\right],\end{equation} which is the smallest integer  such that $L\ge \l\ge L^{\delta^{-s}} $ (here $[\cdot ]$ denotes the integer part).
Bounding each time $\l_k$ by $L^{1/\delta}$, we obtain this way a sequence of $\l_k,t_k, a_k$ with 
$$K_{\l +L^\beta-2 s L^{\gamma/\delta}}(a) \subset K_{t_s}(a_s)\subset K_{\l_s}(a_s) \subset \dots \subset  K_{t_1}(a_1) \subset K_{\l_1}(a_1) \subset K_L(a)  .$$
Choosing $1> \beta>\gamma/\delta$ (which is possible since $\gamma<1$ and $\delta>1$) in view of the definition of $s$, we have   
$L^\beta-2 s L^{\gamma/\delta}\ge 0$ if $L$ is large enough, which ensures that 
$K_\l(a) \subset K_{t_s}(a_s)$.
 
The result of the final step of applying Lemma \ref{lemcenergygb} yields 
 $t_s$ and $a_s$ such that 
\eqref{h2boundary} holds on $\p K_{t_s}(a_s)$ and
\begin{equation}\label{preconc1}
	\left\vb\frac{W(\j_\varphi, \1_{K_{t_s}(a_s)} )}{\vb K_{t_s}(a_s)\vb}-\sigma_m^*\right\vb\le o(1)_{\l\to+\infty}
\end{equation}
and
\begin{align}\label{preconc2}
\left\vb\frac{\nu(K_{t_s}(a_s))}{\vb K_{t_s}(a_s)\vb}-m\right\vb\le \frac{C} {\l^{\gamma}}.
\end{align}
	 Applying one last time Lemma \ref{lemcenergygb} in  $K_{t_s}(a_s)$ (on which the assumptions are satisfied) which contains $K_\l(a)$, more precisely applying \eqref{eqcontrolenergy1}  and combining with \eqref{preconc1}, we obtain  \eqref{eqenergycb} under the assumption $d (K_\l(a), \p K_L) \ge 3  \l^\gamma$, implied by $3\l^\gamma\le L^\beta$ (note that this includes in particular $\l=\bar{c}$ if $L$ is large enough.) 
 \\

 Finally, let us  remove the assumption $d (K_\l(a), \p K_L) \ge 3  \l^\gamma$ by treating the case $3\l^\gamma\ge L^\beta$ 
  and $d(K_\l(a), \p K_L) \ge L^\beta$. In that case, let us   partition $K_{\l} (a)$ into identical squares $\bar{K_i}$ of sidelengths $\in [\bar{c}, 2\bar{c}]$ and satisfying $d(K_i,\p K_L) \ge L^\beta$. Let us then consider $K_i$ the open squares with same centers and sidelength augmented by $1$. They make an open cover of $K_{\l}(a)$.   Let $\chi_i$ be an associated  partition of unity. We note that $\sum_i \chi_i \chi_{K_\l(a)}=\chi_{K_\l(a)}$. We may also require that each $\chi_i$ satisfies \eqref{defcutoff} relative to each $K_i$, in other words, each $K_i$ is a $\chi_{K_i}$. Then also $\chi_{K_\l(a)} \chi_i$ are equal to $\chi_{K_i'}$ for  $K_i'= K_i\cap K_\l(a)$. Since $W(\j, \chi)$ is linear in $\chi$,  we may write 
 $$W(\j_\vp, \chi_{K_\l(a) })= \sum_i W(\j_\vp, \chi_{K_\l(a)}  \chi_i)= \sum_i W(\j_\vp, \chi_{K_i'}).$$
Inserting \eqref{eqenergycb} which is known to hold for the $K_i'$, we obtain \eqref{eqenergycb} for $K_\l(a)$, and this completes this step.}
 \\
 
 \noindent 
 {\bf Step 2.} Proof of \eqref{eqnumpointscb}.
 First, we note that the result of the previous step applied to $\l = \bar{c}$ implies that there exists  a constant $C>0$ (depending on $m$, $M$, $p$, and  $\gamma$) such that 
 \begin{equation}\label{bornlocw} W(\j_\vp, \chi_{K_{\bar{c}  }  (a)} ) \le C,\end{equation} for any $a$ such that $d(a, \p K_L) \ge CL^\gamma$.
 It also follows that, modifying $\bar{c}$ and $C$ if necessary, we  have
 \begin{equation}\label{bornlocnu}
| \nu(K_{\bar{c}}(a))|\le C,\end{equation} for any $K_{\bar{c}(a)}$ satisfying the same assumption.
 To see that it suffices to apply the last step of the bootstrap above with $\l_s= \hal \bar{c}$, then \eqref{preconc2} and the positivity of the measure $\nu$ allow to deduce \eqref{bornlocnu}.
 Finally, combining  these two  facts and Lemma~\ref{lemnormp}, it follows that 
 \begin{equation}\label{norml1}
 \int_{K_{\bar{c}/2} (a) } |\j_\vp|\le C\end{equation}
 for some other constant depending only on $m$, $M$, $p$, $\gamma$, and squares satisfying the same assumption.  
 
 Let now $ K_\l(a)$ be any square satisfying the requirements of the theorem, with $\l$ large enough.
 Let $\chi_{int}$ be a smooth nonnegative function supported in $K_\l(a)$, and equal to $1$ in $K_{\l-1}(a)$, with $\|\nab \chi_{int}\|_{L^\infty} \le 2$. Similarly, let $\chi_{out}$ be a smooth nonnegative function supported in $K_{\l+1}(a)$ and equal to $1$ in $K_\l (a)$ with $\|\nab \chi_{out}\|_{L^\infty} \le 2$. By positivity of $\nu$, we have 
 $  \int \chi_{int} \nu\le \nu(K_\l(a)) \le \int \chi_{out} \nu$, and also by boundedness of $m$ and definition of $\chi_{int}$ and $\chi_{out}$ we have 
 \begin{equation}\label{eavlnu}
 \int \chi_{int} (\nu - m) - C \l   \le \nu (K_\l(a))- m|K_\l(a)| \le \int \chi_{out} (\nu -m) + C\l\end{equation}
 where $C$ depends only on $m$.
 On the other hand, using that $\div \j_\vp = 2\pi (\nu-m)$ and integrating by parts, we have 
 $$\left| \int \chi (\nu - m)\right|\le \frac{1}{2\pi} \int |\nab \chi| |\j_\vp|,$$
 where $\chi$ stands for $\chi_{int} $ or $\chi_{out}$. 
 But the support of $\nab \chi$ can be split into $O(\l)$ squares of size $\bar{c}/2$, thus on which \eqref{norml1} holds. It thus follows that $\left| \int \chi (\nu - m)\right|\le C \l$, for some other constant $C>0$ depending only on $m$, $M$, $p$, $\gamma$,   and inserting into \eqref{eavlnu}, the result 
 \eqref{eqnumpointscb} follows.
\\

\noindent 
{\bf Step 3.} Case where $\int_{\p K_L\cap K_\l(b)} |\vp|^p \le M \l^{2-\gamma}.$
{ First, we may always extend the vector field $\j_\vp$ outside of $K_L$ to $K_{L+L^\g}$, using  multiple times Proposition \ref{propj+j-} with $\ro=m$.  This gives a vector field $\j_\vp$ satisfying the same a priori bounds. Let $K_\l(a)\subset K_L$. Just as in Step 1, we define  a finite sequence $\l_k$ with $\l_k=\max(\l+ 3\l^\gamma, \l_{k-1}^{1/\delta}$) with terminates at $ \l + 2\l^\gamma$. 

We first find $a_1$ such that  $$K_{\l + L^\beta} (a) \subset K_{\l_1 } (a_1) \subset K_{L+L^\gamma}.$$
We may then apply Lemma~\ref{lemcenergygb} to the extended $\j_\vp$ in $K_{\l_1}(a_1)$. It provides a $K_{t_1}(a_1)$  with $ \l_1 - 2\l_1^\gamma\le t_1 \le \l_1- \l_1^\gamma$, and 
\begin{equation}\label{bonboragain}
\int_{\p K_{t_1}(a_1)} |\j_\vp |^p \le M \l_1^{2-\gamma}\end{equation} and \eqref{nukt1}. If $s\ge 2$,  we may next find $a_2$ such that
$$  K_{\l + L^\beta - 2\l_1^{\gamma}} (a) \subset K_{\l_2}(a_2)   \subset K_{t_1}(a_1) \subset K_{\l_1 } (a_1) \subset K_{L+L^\gamma}.$$
We then consider $R_{t_1}= K_{t_1(a_1)}\cap K_L$.
It is a rectangle and its sidelengths are both comparable to $ \l_1$.
We also note that $R_{t_1}$ contains $K_{\l_2}(a_2)\cap K_L$ which contains $K_\l(a)$. 
 Since $\l_1 \ge \l$ and because of the additional assumption placed on $\vp$ on $\p K_L$, and of \eqref{bonboragain}, we have \begin{equation}\label{bonboragain2}
\int_{\p R_{t_1}} |\j_\vp |^p \le M \l_1^{2-\gamma}\end{equation}
We may then apply Lemma \ref{propestimationgb} or rather Remark \ref{aprelem31}, which yields that 
\begin{equation}\label{wr1}
\left\vb\frac{W(\j_\varphi,\1_{R_{t_1}} )}{\vb R_{t_1}\vb}-\sigma_m^*\right\vb\le o(1)_{\l_1\to+\infty}.
\end{equation}
We now iterate the reasoning: we consider $\j_\vp$ restricted to $R_{t_1}$ and extend it outside $R_{t_1}$ using Proposition \ref{propj+j-} with $\ro=m$.  This gives again a vector field $\j_\vp$ satisfying the same a priori bounds.
We can then reapply Lemma \ref{lemcenergygb} to the extended vector field, with respect to the square $K_{\l_2(a_2)}$, the assumptions being verified with the same constants. This gives a $t_2\ge \l_2-2\l_2^\gamma$, from which we define an $R_{t_2}$, etc.
The last iteration of the reasoning gives $t_s\in [\l_s-2\l_s^\gamma, \l_s - \l_s^\gamma]\subset[\l, \l+3\l^\gamma]$ and a rectangle $R_{t_s}$ containing $ K_\l(a)$ (hence its sidelengths have to be  $\sim\l$) on which 
\begin{equation}\label{wrs}
\left\vb\frac{W(\j_\varphi,\1_{R_{t_s}} )}{\vb R_{t_s}\vb}-\sigma_m^*\right\vb\le o(1)_{\l\to+\infty}.
\end{equation} and 
\begin{align}\label{nukts}
\left\vb\frac{\nu(R_{t_s})}{\vb R_{t_s} \vb}-m\right\vb\le \frac{C} {\l^{\gamma}}
\end{align}
Applying then Lemma \ref{lemcenergygb} over $R_{t_s}$, the conclusion \eqref{eqcontrolenergychi1} provides the desired result. 
If the assumption holds with $\l $ of order $1$, then we can conclude also that 
 \eqref{eqnumpointscb} holds, by arguing exactly as in Step 2.}

\subsection{Proof of Theorem \ref{thmmainper}}

First of all, we observe that \eqref{eqlimsigmaper} is proved in Proposition \ref{proplimmincon}.



With the same arguments used in the proof of Lemma \ref{propestimationgb}, we may  obtain the following lemma.
\begin{lem}\label{propestimationgbper}
Let $m$ be a positive number. Let  $L>1$ and $\mathbb T_{L}:=\RR^2/(2L\mathbb Z)^2$ such that $\vb \mathbb T_L\vb\in \mathbb N$.
Let
 $\bj$ be a minimizer of $W(\j,\1_{\TT_{L}})$ on $\mathcal{A}_{m,per}(\TT_L)$.
For every $a\in \TT_L$, let $\psi$ be the restriction of $\bj\cdot \vnu$ to $\partial K_\l(a)$.

If $\psi$ satisfies \eqref{h2boundary} in $K_\l(a)$ for some $p\in(1,2)$, $\gamma\in\left(\frac{3-p}{2},1\right)$ and $M>0$, then
\begin{equation}
\left\vb\frac{W(\bj,\1_{K_{l}(a)} )}{\vb K_{\l(a)}\vb}-\sigma_m^*\right\vb\le o(1)_{\l\to+\infty}.
\end{equation}
\end{lem}

To conclude the proof of the theorem, we proceed as in the proof of Theorem \ref{thmmaincb}.
Let $\bar\j$ be a minimizer for $\sigma_{per}(L;m)$; then $\bj$ is $K_{L}(a)$-periodic and $W(\bar \j)=\sigma_{per}(L;m)$

As a consequence of \eqref{eqlimsigmaper} and of the definition of $W$,
$$
\frac{W(\bar\j,\chi_{K_L})}{\vb \mathbb T_{L}\vb}\le W(\bar \j)+o(1)_{L\to+\infty}\le \sigma^*_m+o(1)_{L\to+\infty}.
$$
Moreover,
$$
\frac{\nu(K_{L+1})}{\vb K_L\vb}\le 2m
$$
Hence, if $L$ is large enough we can take $C_1=2m$ and $C_2=\sigma^*_m+1$ and apply Lemma \ref{lemgoodboundary}. Then we may apply Lemma \ref{lemcenergygb} and iterate as in the proof of Theorem~\ref{thmmaincb} to obtain \eqref{eqcontrolenergychi1} and  conclude the proof. The proof of \eqref{eqnumpointsper}  is the same as above in Theorem \ref{thmmaincb}.


\subsection{Proof of Lemmas \ref{lemgoodboundary} and \ref{lemcenergygb}}

The proof of Lemma \ref{lemgoodboundary} is an adaptation of the proof of Lemma 4.14 of \cite{sandierserfaty}.
{\begin{proof}[Proof of Lemma \ref{lemgoodboundary}]
Let $a\in K_L$ and $\l$ such that (\ref{eqconddelta}) is satisfied.

\noindent\textbf{Step 1:} Denote by $g_L$ the result of applying Proposition \ref{prop49ss}  in  $K_{L+1}$ to $(\j,\nu)$. We apply \eqref{eqestimWg} to functions of the form  $\chi(x)=\vartheta(\norm x-a\norm_\infty)$, i.e. whose level sets are squares centered in $a$, with the additional assumption that $\vartheta'(t)=0$ outside $[r-2,r-1]$ and $\vartheta=0$ on $[r-1,+\infty)$ with $r\le \l-3$. Since for any Radon measure $\mu$ on $K_L$ we have
$$
\int\chi\,d\mu=-\int_0^{r-1}\vartheta'(t)\mu(K_t{(a)})\,dt,
$$
we deduce with \eqref{eqestimWg} that
\begin{align}\label{eqestimWglrho}
\int_{r-2}^{r-1}(W(\j,\1_{K_t{(a)}})&-g_L(K_t{(a)}))\vartheta'(t)\,dt\nonumber\\
&=-W(\j,\chi)+\int \chi\,dg_L\le  Cn(\log n+\norm \rho\norm_{L^\infty(K_L)})\norm \vartheta'\norm_\infty
\end{align} in view of \eqref{eqestimWg}
where $n=\#\{p\in\Lambda\vb B(p,1)\cap\mathrm{supp}(\nabla\chi)\neq\emptyset\}$, so that 
\begin{equation}\label{eqestimn}
 n\le \nu(K_{r+1}{(a)})-\nu(K_{r-2}{(a)})
\le C \l^{1-\frac{1}{p}} L^{\frac{2}{p}} \log^{\hal} L + C \l \le C
 \l^{\frac{2\delta- 1}{p}+1} \log^{\frac{1}{2}}  \l
\end{equation}where we have used Lemma \ref{lembor} and  $L\le \l^\delta$. Here 
 $ C$ depends only on $p$, $\norm \rho\norm_{L^\infty}$, and on the constants in (\ref{eqcondnusquare}) and (\ref{eqcondenergy}).
Inserting this into \eqref{eqestimWglrho}, we deduce by duality that
\begin{equation}\label{eqestimWgl2}
\int_{r-2}^{r-1}\vb W(\j,\1_{K_t{(a)}})-g_L(K_t{(a)})\vb\,dt\le  C\l^{\frac{2\delta -1}{p}+1}\log^{\frac{3}{2}}\l.
\end{equation}

\noindent\textbf{Step 2:} For any integer $k\ge 1$ let $\xi_{k}=\chi_{K_{k+1}{(a)}}-\chi_{K_k{(a)}}$, and let $\xi_0=\chi_{K_1(a)}$. Then $\xi_{k}\ge 0$, since $\chi_{K_{k+1}{(a)}}=1$ on $K_k{(a)}$ and $\chi_{K_k{(a)}}\le 1$ and is supported in $K_k{(a)}$. Moreover $\xi_{k}$ is supported in $\mathcal C_k=K_{k+1}{(a)}\backslash K_{k-1}{(a)}$. Since (\ref{eqcondnusquare}) holds and $\l \le L^\delta$, the number of integers $k$ in $[\l-2\l^{\gamma}+2,\l-\l^{\gamma}-2]$ such that $\nu(K_{k+2}{(a)}\backslash K_{k-2}{(a)})\le \tilde C\l^{2\delta-\gamma}$ is greater than $\frac{\l^\gamma}{2}$ if $\tilde C$ is chosen large enough. On the other hand, using $g_L \ge -C$, we have
\begin{multline*}
\sum_{k=[\l-2\l^{\gamma}+2]}^{[\l-\l^{\gamma}-2]}\int \xi_{k}\,dg_L=\int(\chi_{K_{[\l-\l^{\gamma}]-1}{(a)}}-\chi_{K_{[\l-2\l^{\gamma}+2]}{(a)}})
\,dg_L
\le \int \chi_{K_L} \, dg_L + C L^2.
\end{multline*}
We then  observe that $|\int \chi_{K_L} \, dg_L- W(\j, \chi_{K_L}) |
\le L^2 { \log L}$ using \eqref{eqestimWg}
 and \eqref{eqcondnusquare} to bound $n\log n$. Since $W(\j, \chi_{K_L})\le C_2 |K_L|$ by \eqref{eqcondenergy} and $L\le \l^\delta$, it follows that
$$\sum_{k=[\l-2\l^{\gamma}+2]}^{[\l-\l^{\gamma}-2]}\int \xi_{k}\,dg_L\le
  CL^2{\log L} \le  C \l^{2\delta}{\log \l }.
$$
 Since $g_L\ge -C$ we have $\int \xi_{k}\,dg_L\ge - C\l$ and therefore the number of integer $k$'s between $[\l-2\l^{\gamma}+2]$ and $[\l-\l^{\gamma}-2]$ such that $\int \xi_{k}\,dg_L\le \tilde C\l^{2\delta-\gamma}{ \log \l}$ is larger than $\frac{\l^{\gamma}}{2}$ if $\tilde C$ and $\l$ are chosen large enough. We can thus choose an integer $k\in [\l-2\l^{\gamma}+2,\l-\l^{\gamma}-2]$ satisfying both conditions, i.e.
\begin{equation}\label{eqcontrolCk}
\nu(K_{k+2}{(a)}\backslash K_{k-2}{(a)})\le \tilde C\l^{2\delta-\gamma}, \quad\int \xi_{k}\,dg_L\le \tilde C\l^{2\delta-\gamma}{\log \l},
\end{equation}
for some $\tilde C$ which depends on $C_1$ and $C_2$.

Applying Proposition \ref{prop49ss} in $\mathcal C_k$ to $\xi_{k}$,  and using \eqref{eqcontrolCk} to control $n\log n$, we deduce
\begin{equation*}
\left\vb W(\j,\xi_{k})-\int \xi_{k}\,dg_L\right\vb\le
{ C\l^{2\delta -\gamma}\log \l}
\end{equation*}
hence  $W(\j,\xi_{k})\le   C\l^{2\delta-\gamma}\log \l$.
Applying Lemma \ref{lemnormp} over $\mathcal C_k$, we find for $p<2$
$$
\int_{\mathcal C_k}\vb \xi_{k}\vb^{\frac{p}{2}}\vb \j\vb^p\le \tilde C\l^{1-\frac{p}{2}}(\l^{2\delta-\gamma}\log \l)^{p/2}\le  C \l^{2-\gamma}.
$$
 because $\delta< \frac{1}{p}\left(1-\gamma+\frac{p}{2}(1+\gamma)\right)$. Since $\xi_{k}=1$ if $\norm x-{ a}\norm_{\infty}=k$, and $|\nab \xi_k |\le C$, it follows that
$$
\int_{K_{k+\frac{1}{ C}}{(a)}\backslash K_{k-\frac{1}{ C}}{(a)}}\vb \j\vb^p\le  C \l^{2-\gamma}.
$$
By a mean value argument on this integral as well as on \eqref{eqestimWgl2} (applied to $r=k+1$), 
we deduce the existence of $t\in[k-1,k]$, hence $t\in [\l-2\l^{\gamma},\l-\l^{\gamma}]$, such that, on the one hand
\begin{equation*}
\int_{\partial K_t{(a)}}\vb \j\vb^p\le  C \l^{2-\gamma}
\end{equation*}
proving \eqref{eqgoodboundary};  and on the other hand
\begin{equation}\label{eqestimWgl3}
\vb W(\j,\1_{K_t{(a)}})-g_L(K_t{(a)})\vb\le  C\l^{\frac{2\delta -1}{p}+1}\log^{\frac{3}{2}}\l.
\end{equation}
Next, using again that $g_L\ge  - C$ we have that
$$
 g_L(K_t(a))\le \int \chi_{K_{\l}(a)} \,dg_L + C \l^{1+\gamma}.$$
Combining with  \eqref{eqestimWg} and using  \eqref{eqestimn} to control the error,  we are led to
$$
g_{L}(K_t{(a)})\le W(\j,\chi_{K_\l{(a)}})+ C \l^{1+\gamma}+ C\l^{\frac{2\delta -1}{p}+1}\log^{\frac{3}{2}}\l,
$$
which together with \eqref{eqestimWgl3} yields \eqref{eqgoodboundarypropenergy}. 
Finally, from \eqref{eqgoodboundary} and H\"older's inequality, we have
\begin{align*}
\left\vb \nu(K_t{(a)})-\int_{K_t(a)}\rho\,dx\right\vb=&\,\left\vb\int_{\partial K_t}\j\cdot\vnu\right\vb\le\norm \j \norm_{L^p(\partial K_t)}\vb\partial K_t\vb^{1-\frac{1}{p}}\\
\le&\, C \l^{\frac{2-\gamma}{p} + 1-\frac{1}{p}}\le C \l^{2-\gamma}
\end{align*}
since $\frac{2-\gamma}{p} + 1-\frac{1}{p}<2-\gamma$ whenever $p>1$.
\end{proof}}


\begin{proof}[Proof of Lemma \ref{lemcenergygb}] 
First of all, we apply Proposition \ref{propj+j-} to $\vp$ with $\rho(x)=1$ in $ K_{L+2L^\a}\backslash K_L$; this gives us a vector field $\j_+$, through which we  can extend $\j$  into a vector field, still denoted  $\j$, on $K_{L+L^\a}$. It satisfies,
 for all $r\in[L+1,L+L^\a]$,
\begin{equation*}
 	W(\j,\chi_{K_{r}})\le W(\j,\1_{K_{L}})+ C L^{1+\beta}
 \end{equation*}
 and
\begin{equation}\label{eqnumpointextr}
	\nu(K_{r})\le \nu(K_{L})+ CL^{1+ \gamma}
\end{equation}
for some positive constant $ C$ and for some $\beta \in (0,1)$, depending on $p$, $M$, $\gamma$. Next, let $g_{L}$  be the result of applying Proposition \ref{prop49ss} in $K_{L+1}$; by using \eqref{eqestimWg}, we have for $a\in K_L$ such that $K_\l(a)\subset K_L$,
\begin{multline}\label{dete}
 W(\j,\chi_{K_{\l}(a)})= W(\j,\chi_{K_{L+1}})-W(\j,\chi_{K_{L+1}}-\chi_{K_{\l}(a)})\\
  \le  W(\j,\chi_{K_{L+1}})-\int   (\chi_{K_{L+1}}-\chi_{K_{\l}(a)})  \,   dg_{L} + Cn_{L+1}\log n_{L+1} +Cn_\l\log n_\l
\end{multline}
with $n_{L+1}\le \nu(K_{L+2})-\nu(K_{L-1})$ and $n_\l\le\nu(K_{\l+1}{(a)})-\nu(K_{\l-2}{(a)})$.
In view of Lemma \ref{lembor} applied in $K_{L+4}$, we have
\begin{align*}  
  n_{L+1}\le C L^{1+\frac{1}{p}}\log^{\frac{1}{2}}L\ \mbox{and}\ n_\l\le  C \l^{1-\frac{1}{p}}L^{\frac{2}{p}}\log^{\frac{1}{2}}L\le C L^{1+\frac{1}{p}}\log^{\frac{1}{2}}L.
\end{align*}
Inserting into \eqref{dete} and using that $g_L\ge -C$ it follows that
\begin{align*}
  W(\j,\chi_{K_{\l}(a)})
  \le&\, W(\j,\chi_{K_{L+1}})+C\vb K_{L+1}\backslash K_{\l-1}(a)\vb+ C L^{1+\frac{1}{p}}\log^{\frac{3}{2}}L
\end{align*}
which gives \eqref{eqcontrolenergy1}, as well as \eqref{eqcontrolenergychi1} by choosing $K_l(a)=K_L$.

Finally, taking $r=L+1$ in \eqref{eqnumpointextr} and using \eqref{eqcontrolenergychi1}, we find that  the assumptions of Lemma \ref{lemgoodboundary} are satisfied in  $K_{L+1}$ for the extended $\j$. We may then obtain the same results.

\end{proof}

\section{The Case of a Non-Constant Background}\label{sectionnonconstant}
The goal of this section is to obtain similar results as Proposition \ref{proplimmincon}, but in the case of a varying background, in preparation for the study of the Coulomb gas minimizers.

The proofs are   similar to the above and also to  those of \cite[Section 7]{sandierserfatygas}, except that we have to be more careful with error terms due to the more general varying background. We outline the main differences.

\begin{prop} \label{proplimminrhoup} Let $a\in\mathbb R^2$, $L>0$ and $\frac{1}{2}\le\lambda\le 1$. Let $\rho$ be a nonnegative $\mathcal C^{0,\lambda}(K_L(a))$ function for which there exist $\underline \rho, \overline \rho>0$ such that $\underline \rho \le \rho(x)\le \overline \rho$. 
If $\int_{K_L(a)} \rho(x)\,dx \in \NN$, we have
\begin{align}\label{eqlimsigmarhoup}
\sigma_{0}(K_{L}(a); \rho)\le& \frac{1}{\vb K_L(a)\vb}\int_{K_L(a)}\min_{\mathcal A_{\rho(x)}}W\,dx+C\left(\norm \rho\norm_{\mathcal C^{0,\lambda}(K_L(a))}L^{\beta \lambda}\right)\\
&+C\left(\norm \rho\norm_{\mathcal C^{0,\lambda}(K_L(a))}L^{\beta \lambda}\right)^2+o(1)_{L\to+\infty}\nonumber
\end{align}
for some $\beta\in(0,1)$ and $C$ positive constant.
\end{prop}

\begin{proof} Since the proof is very similar to \cite[Section 7]{sandierserfatygas}, although in a simpler setting,  we only sketch the main steps.\\
\noindent
{\bf Step 1.} 
We choose $\alpha<\frac{\lambda}{1+\lambda}-\varepsilon$.
We claim that if $L$ is  large enough,  we can construct  a collection $\mathcal K$ of rectangles which partition  $K_L(a)$,  whose sidelengths are between $2L^{\alpha}-O\left(\frac{1}{L^{\alpha}}\right)$
and $2L^{\alpha}+O\left(\frac{1}{L^{\alpha}}\right)$, and such that for all $K\in \mathcal K$ we have
$
\int_K\rho(x)\,dx\in \NN.
$
To show this, it suffices to proceed as in Step 1 of the proof of Proposition \ref{propj+j-}, i.e. cutting first $K_L(a)$ into horizontal strips  of width $\sim  L^{\alpha}$ in which $\int \ro$ is an integer, and then cutting each strip vertically into rectangles in which $\int \ro$ is again an integer.

We then set $\l=L^\alpha$.\\
\noindent 
{\bf Step 2.} We denote by $x_K$ the center of each $K$ and $\rho_{K}=\fint_{K}\rho(x)\,dx$.
Using \cite[Proposition 7.4]{sandierserfatygas}, which allows to truncate a minimizer of $W$ into a given rectangle,  and rescaling the obtained vector-field by  $\sqrt{\ro_K}$ we obtain in each $K\in \mathcal K$ a vector-field $E_K$ satisfying
$$
\left\{\begin{aligned}&\di \j_K=2\pi\Big(\sum_{p\in \Lambda_K}\delta_p -\ro_ K \Big)&\mbox{in } K\\ &\j_K\cdot \vnu=0 &\mbox{on }\partial  K \end{aligned}\right.;
$$
for some discrete subset $ \Lambda_K\in  K$, and
\begin{equation}\label{eqestimenergy1}
\frac{W(\j_K,\1_{ K})}{\vb  K\vb}\le \min_{\mathcal A_{\ro_K}}W+o(1)_{\l\to +\infty}.
\end{equation}
Using Lemma \ref{lemnormp} after extending  for example $\j_K$ via Proposition \ref{propj+j-} we also control the $L^p$ norm of $\j_K$ for $p\in (1,2)$:
\begin{equation}\label{lpnormjk}
\|\j_K\|_{L^p(K)} \le C\l^{\frac{2}{p}}\log^{\hal}\l.
\end{equation}
Then, we have to rectify the weight $\rho_K$. For $K\in \mathcal K$, we let $H _K$ solve $-\Delta H_K=2\pi(\rho_K-\rho)$ on $K$ and $\nab H_K\cdot \vnu=0$ on $\partial K$. Using Lemma \ref{lemsysum} we have for any $q>1$ that
\begin{equation}\label{eqestimweight}
\norm \nabla H_K\norm_{L^q(K)}\le C\l^{1+2/q} \norm \rho-\rho_K\norm _{L^\infty(K)}\le C\l^{1+2/q+\lambda} \|\ro\|_{\mathcal C^{0,\lambda} } .
\end{equation}
We then define $\j$ to be $\j_K - \nab H_K$ in each $K\in \mathcal{K}$. Pasting these together  defines a  $\j$ over the whole $K_L(a)$, satisfying
$$
\left\{\begin{aligned}&\di \j=2\pi\Big(\sum_{p\in \Lambda}\delta_p -\ro \Big)&\mbox{in } K_L(a)\\ &\j\cdot \vnu=0 &\mbox{on }\partial  K_L(a) \end{aligned}\right.;$$ for some discrete set $\Lambda$,
since no divergence is created at the boundaries between the $K$'s, and such that $\curl \j=0$ near $\Lambda$.
Next we evaluate $W(\j, \indic_{K_L(a)})$.  The control follows from \eqref{eqestimenergy1}, \eqref{lpnormjk} and \eqref{eqestimweight}, using Lemma \ref{lemsumvf}: for $p\in (1,2)$ and $1/p+1/q=1$, we find
\begin{multline*}
W(E,\indic_K) \le |K|\min_{\mathcal{A}_{\ro_K}} W + |K|o_{\l\to \infty}(1)
+ C \l^{4+2\lambda} \|\ro\|_{\mathcal C^{0,\lambda}}^2 + C \l^{3+\lambda} \log^{1/2}\l  \|\ro\|_{\mathcal C^{0,\lambda} }
\\
\le |K|\min_{\mathcal{A}_{\ro_K}} W + |K|o_{\l\to \infty}(1)
+ L^{\alpha(4+2\lambda) } \|\ro\|_{\mathcal C^{0,\lambda}}^2+ C L^{\alpha( 3+ \lambda)} \log^{1/2} L  \|\ro\|_{\mathcal C^{0,\lambda} }.\end{multline*}
Summing over all squares (there are $L^{2-2\alpha}$ of them), and using the  H\"older continuity of $\min_{\mathcal A_m}W$ as a function of $m$ (in view of \eqref{eqscalingminW}) we find
\begin{multline*}
W(E,\indic_{K_L(a)})\le \int_{K_L(a)} \min_{\mathcal{A}_{\ro(x)} } W\, dx
+ o_{L\to \infty}(L^2) \\+ CL^{2+2\alpha + 2\alpha \lambda} \|\ro\|_{\mathcal C^{0,\lambda}}^2 +  CL^{2+\alpha + \alpha \lambda}\log^{1/2} L  \|\ro\|_{\mathcal C^{0,\lambda} }. \end{multline*}
Since  $\hal \le \lambda \le 1 $ and we chose  $\alpha$ such that $\alpha + \alpha \lambda<\lambda (1- \ep)$, after dividing by $L^2$ we have exponents  $2\alpha + 2\alpha \lambda <2\beta \lambda$ and $\alpha+ \alpha\lambda <\beta\lambda$ for some $\beta<1$.

To conclude we apply Lemma \ref{lemdiv0} to $\j$ to obtain a new  vector field in $\mathcal B_{\rho,0}(K_L(a))$ while decreasing $W(\j ,\1_{K_L(a)})$.
 The result follows.
 \end{proof}We next state a lemma that allows to reduce to the situation of a constant background density, modulo some error terms.
  \begin{lem}
  \label{lemequivsmallsquare}
   Let $p\in(1,2)$, $\frac{1}{2}\le\lambda\le 1$, $a\in \RR^2 $, $\l>0$ and 
   $\varphi\in L^p(\partial K_\l(a))$.
    Let $\rho$ be a $\mathcal C^{0,\lambda}(K_\l(a))$ function for which there exist $\underline \rho, \overline \rho>0$ such that $\underline \rho \le \rho(x)\le \overline \rho$.Then, for all $\varepsilon>0$,
\begin{align}\label{eqestimsigmarho}
\vb\sigma_{\varphi}(K_{\l}(a); \rho)- \sigma_{\varphi}(K_{\l}(a); \rho_\l)\vb\le C\left(\norm \rho\norm_{\mathcal C^{0,\lambda}(K_{\l}(a))}\l^{\lambda+1+\varepsilon}+(\norm \rho\norm_{\mathcal C^{0,\lambda}(K_{\l}(a))}\l^{\lambda+1})^2\right)
\end{align}
where $\rho_\l=\fint_{K_\l(a)} \rho(x)\,dx>0$.
\end{lem}

\begin{proof}
The proof follows the same ideas as the previous one.
Let $\j$ a minimizer of $\sigma_\varphi(K_\l(a);\rho_l)$; hence satisfying
$$
\left\{\begin{aligned}&\di \j=2\pi(\sum_{p\in\Lambda}\delta_p-\rho_\l)& &\mbox{in } K_\l(a)\\&\j\cdot \vnu=\varphi& &\mbox{on }\partial K_\l(a)\end{aligned}\right.
$$
and $\frac{W(\j,\1_{K_\l(a)})}{\vb K_\l(a)\vb}=\sigma_\varphi(K_\l(a);\rho_\l)$. Next, let $H$ solve
$$
\left\{\begin{aligned}-\Delta H & =2\pi(\rho_\l-\rho)& &\mbox{in } K_\l(a)\\ \nab H \cdot{\vnu} & =0& &\mbox{on }\partial K_\l(a)\end{aligned}\right..
$$
By Lemma \ref{lemsysum}, for $q>1$
$$
\norm \nabla H\norm_{L^q(K_\l(a))}\le C \l^{1+2/q} \norm \rho_\l-\rho\norm_{L^{\infty}(K_l(a))}\le C \l^{1+2/q+\lambda} \|\ro\|_{\mathcal{C}^{0,\lambda}}.
$$
Now, let $\tj=\j-\nabla H$. Using  again Lemma \ref{lemsumvf} and the same arguments as in the previous proof, we find
$$
\frac{W(\tj,\1_{K_\l(a)})}{\vb K_\l(a)\vb}\le \frac{W(\j,\1_{K_\l(a)})}{\vb K_\l(a)\vb}+C\left(\norm \rho\norm_{\mathcal C^{0,\lambda}}\l^{\lambda+1+\varepsilon}+(\norm \rho\norm_{\mathcal C^{0,\lambda}}\l^{\lambda+1})^2\right).
$$
Finally, we conclude by modifying $\tj$ to have $\curl \tj=0$ using Lemma \ref{lemdiv0}; this way  $\tj\in \mathcal B_{\rho,\varphi}(K_L(a))$ and
$$
\sigma_{\varphi}(K_{\l}(a); \rho)\le\frac{W(\tj,\1_{K_\l(a)})}{\vb K_\l(a)\vb}\le \sigma_{\varphi}(K_{\l}(a); \rho_\l)+C\left(\norm \rho\norm_{\mathcal C^{0,\lambda}}\l^{\lambda+1+\varepsilon}+(\norm \rho\norm_{\mathcal C^{0,\lambda}}\l^{\lambda+1})^2\right).
$$
To obtain the other inequality, it suffices to take $\oj=\j+\nabla H$ with $H$ defined as above and $\j$ such that  $\frac{W(\j,\1_{K_\l(a)})}{\vb K_\l(a)\vb}=\sigma_\varphi(K_\l(a);\rho)$.
\end{proof}

\begin{prop} \label{proplimminrholow} Let $p\in(1,2)$, $\frac{1}{2}\le\lambda\le 1$ and $\rho$ be a $\mathcal C^{0,\lambda}(K_L(a))$ function for which there exist $\underline \rho, \overline \rho>0$ such that $\underline \rho \le \rho(x)\le \overline \rho$.  Suppose $\int_{K_L(a)} \ro(x)\, dx\in \mn$. Then\begin{align}\label{eqlimsigmarholow}
\sigma_{0}(K_{L}(a); \rho)\ge& \frac{1}{\vb K_L(a)\vb}\int_{K_L(a)}\min_{\mathcal A_{\rho(x)}}W\,dx-C\left(\norm \rho\norm_{\mathcal C^{0,\lambda}(K_L(a))}L^{\beta \lambda}\right)\\
&-C\left(\norm \rho\norm_{\mathcal C^{0,\lambda}(K_L(a))}L^{\beta \lambda}\right)^2+o(1)_{L\to+\infty}\nonumber
\end{align}
for some $\beta\in(0,1)$.
\end{prop}

\begin{proof} Consider $\j$ a minimizer for $\sigma_0 (K_L(a); \ro)$. Since $\j \in \mathcal{B}_{\ro, 0}(K_L(a))$ we have   $\nu (K_L(a))= \int_{K_L(a)} \ro \, dx\le \overline\ro |K_L|$.




We  start by extending 
 $\j$  to a slightly bigger square using Proposition \ref{propj+j-},  having extended $\ro$ to a function with same $C^{0,\lambda}$ norm (or at most double),  in such a way that \eqref{energyj+chi} holds and $\nu(K_{L+1}(a))\le ( \overline \ro +1) |K_{L+1}|$.

The proof consists in a combination of a partitioning argument together with the bootstrapping method employed to prove Theorem \ref{thmmaincb}. We let $\alpha$ be as in the proof of Proposition \ref{proplimminrhoup}, and we choose a sequence of lengthscales $\l_k$ with $\l_0=L $, $\l_1 = L^{1/\delta}$,  $\l_s\le L^\a$
and $\l_{k-1}\ge \l_k \ge \l_{k-1}^{1/\delta}$ for $k=1, \dots, s$, where $\delta$ is as in \eqref{eqconddelta} and such that $\delta<\frac{1+p}{2}$. We can do this with $s\le -\frac{ \log \alpha}{\log \delta}$, so in a {\it bounded} number of steps (by contrast to the proof of Theorem \ref{thmmaincb}), and we will reason by iteration on $s$.

First let us partition $K_{L+1}(a)$ into an integer number of identical  squares $\bar{K_i}(a_i)$ of sidelength $\ell_1 $.
Let us then  consider $K_i$ the open squares with same centers $a_i$ but sidelength augmented by $1$. The $K_i$'s are overlapping squares which make an open cover of $\overline{K}_{L+1}(a)$. Let $\chi_i$ be an associated partition of unity.  We may require that each $\chi_i$ satisfies \eqref{defcutoff} relatively to each $K_i$, in other words $\chi_i$ is a $\chi_{K_i}$.
Since $\sum_i \chi_i= 1 $ in $K_L(a)$ and $W(\j, \chi)$ is linear with respect to $\chi$,  we  check that
\begin{equation*}\label{partunit0}
W(\j, \indic_{K_L(a)})= \sum_i  W(\j , \chi_i \indic_{K_L(a)})= \sum_i W(\j,\chi_i) -  W(\j, \indic_{K_L(a)^c}\sum_i \chi_i)
.\end{equation*}
Combining this with \eqref{energyj+chi} it follows that
\begin{equation}\label{partunit}
W(\j, \indic_{K_L(a)})\ge \sum_i W(\j, \chi_i) - C(L^{1+\beta}+L^{1+\gamma})-C  L^2 \|\ro\|_{C^{0,\lambda}}L^{\beta\lambda} (1 + \|\ro\|_{C^{0,\lambda}}L^{\beta\lambda}).\end{equation}

We next turn to bounding from below $\sum_i W(\j, \chi_i)$. 
We may assume that 
\begin{equation}\label{assummaxw}
W(\j ,\chi_i) \le   C_1  |K_i|, 
\end{equation}
for otherwise, we have a lower bound $W(\j, \chi_i)\ge C_1|K_i|$ which will suffice (if $C_1$ is chosen large enough).


Applying Lemma \ref{lemcenergygb} in $K_{L}(a)$, we find that there exists $K_{t_i}(a_i) \subset K_i \cap K_L(a)$ such that $\ell_1- 2\ell_1^\gamma \le t_i \le \ell- \ell_1^\gamma$,
\begin{equation}\label{bonnnnbor}
\int_{\partial K_{t_i}(a_i)}  |\j|^p \le M\ell_1^{2-\gamma},\end{equation} and
\begin{equation}\label{compchi}
W(\j, \chi_i) \ge W(\j, \1_{K_{t_i}(a_i)} ) - C \ell_1^{\frac{2\delta-1}{p} + 1} \log^{3/2} \ell_1 \ge W(\j, \1_{K_{t_i}(a_i)} ) - o(\ell_1^2) ,\end{equation} by choice of $\delta$. Combining with \eqref{assummaxw} it follows that 
\begin{equation}\label{bonborneW}
W(\j, \1_{K_{t_i}(a_i)} )\le (C_1+1) |K_{t_i}|.
\end{equation}
Moreover, from \eqref{bonnnnbor} and if $L$ is large enough,  we deduce that 
$\nu(K_{t_i}) \le (\|\ro\|_{L^\infty}+1) |K_{t_i}|$.
The assumptions of Lemma \ref{lemcenergygb} are thus satisfied again on $K_{t_i}$. Inserting \eqref{compchi} into \eqref{partunit}, we are led to 
\begin{multline}\label{partunit2}
W(\j, \indic_{K_L(a)})\ge \sum_i\min \(C_1 |K_i|, W(\j, \1_{K_{t_i}(a_i)})\) \\- o(L^2)- 
C  L^2 \|\ro\|_{C^{0,\lambda}}L^{\beta\lambda} (1 + \|\ro\|_{C^{0,\lambda}}L^{\beta\lambda}).\end{multline}
We then iterate the reasoning starting from $K_{t_i}(a_i)$, and partitioning it into squares of size $\l_2$,  which themselves get partitioned, etc, until the scale $\l_s\le L^\alpha$. 
Let $K_t(b)$ be one of the squares obtained this way at the scale $\l_s$, and let $\vp$ be the associated boundary data controlled by a relation of the form \eqref{bonnnnbor}. 
In view  Lemma \ref{lemequivsmallsquare} we  have
$$\sigma_{\vp}  (K_{t}(b); \ro) \ge \sigma_{\vp} (K_{t}(b); \tilde{\ro}) - C (\|\ro\|_{\mathcal C^{0,\lambda}} \ell_s^{1+\lambda +\ep}  +   \|\ro\|_{\mathcal C^{0,\lambda} }^2 \ell_s^{2+2\lambda } ),
$$
with $\tilde{\ro}= \fint_{K_{t}(b)} \ro.$
Since $W(E, \1_{K_{t}(b)}) \ge 
 \sigma_{\vp}(K_{t}(b); \ro) |K_{t}|$ and $\l_s\le L^\alpha$,  using Proposition \ref{proplimmincon}, we find
\begin{multline*}
W(\j, \1_{K_t (b)} ) \\
\ge |K_t|\min \(C_1 ,  \min_{\mathcal A_{\tilde \ro}} W  +o(1)
 - C (\|\ro\|_{\mathcal C^{0,\lambda}} L^{\alpha(1+\lambda +\ep)}  +   \|\ro\|_{\mathcal C^{0,\lambda} }^2 L^{\alpha(2+2\lambda)})\).
 \end{multline*} 
  Since $\alpha$ has been chosen as in Proposition 
\ref{proplimminrhoup}, we get in the same way as there that 
$L^{\alpha(1+\lambda +\ep)}\le L^{\beta\lambda}$ and $L^{\alpha(2+2\lambda)}\le L^{2\beta \lambda}$ for some $\beta<1$.
Since also $L^{2\alpha}\le O(|K_{t}|)$  and $\min_{\mathcal{A}_\ro} W$ is continuous with respect to $\ro$, we conclude that
\begin{multline*}
W(\j, \1_{K_t (b)} ) \\
\ge |K_t|\min \(C_1 , 
 \fint_{K_t(b)}  \min_{\mathcal A_{\ro(x)}}  W \, dx
  - C   (\|\ro\|_{\mathcal C^{0,\lambda}} L^{\beta\lambda}  +   \|\ro\|_{\mathcal C^{0,\lambda} }^2 L^{2\beta \lambda}) +o(1) \) .\end{multline*}
  Inserting all the estimates obtained at all these scales, until  $\l_s$ into \eqref{partunit2}, using the fact that the number of steps is bounded, and choosing $C_1 \ge \min_{m \in [\underline \ro, \overline \ro]} \min_{\mathcal A_m} W$,  
   we are led to 
  \begin{multline}\label{partunit3}
W(\j, \indic_{K_L(a)})\ge \int_{K_L(a)}  \min_{\mathcal A_{\ro(x)}}  W \, dx
 \\- o(L^2)- 
C  L^2 \|\ro\|_{C^{0,\lambda}}L^{\beta\lambda} (1 + \|\ro\|_{C^{0,\lambda}}L^{\beta\lambda}).\end{multline}
Since $\j$ is a minimizer for $\sigma_0(K_L(a); \ro)$, this proves the result (after dividing by $L^2$).

\end{proof}

We have the following corollary as a direct consequence of 
Propositions \ref{proplimminrhoup} and  \ref{proplimminrholow}, using  Lemma \ref{corestimationmin}.

\begin{cor} \label{proplimminnc}
Let $p\in(1,2)$, $a\in\mathbb R^2$ and $\frac{1}{2}\le\lambda\le 1$. Let $\rho$ be a nonnegative $\mathcal C^{0,\lambda}(K_L(a))$ function for which there exist $\underline \rho, \overline \rho>0$ such that $\underline \rho \le \rho(x)\le \overline \rho$. 
 If $\norm \rho\norm_{\mathcal C^{0,\lambda}(K_L(a))}L^{\beta \lambda}=o(1)_{L\to +\infty}$ for all $\beta \in (0,1)$, then
\begin{enumerate}
\item for all sequences of real numbers $L$ such that $\int_{K_L(a)}\rho(x)\,dx \in \NN$, we have
\begin{align}\label{eqlimsigma0nc}
\lim\limits_{L\to+\infty} \sigma_{0}(K_L(a);\rho)- \frac{1}{\vb K_L(a)\vb}\int_{K_L(a)}\min_{\mathcal A_{\rho(x)}}W\,dx=0;
\end{align}
\item  given $\gamma\in\left(\frac{3-p}{2},1\right)$ and $M>0$, we have
\begin{align}\label{eqlimsigmaphinc}
\lim\limits_{L\to+\infty} \sigma_{\varphi}(K_L(a);\rho)- \frac{1}{\vb K_L(a)\vb}\int_{K_L(a)}\min_{\mathcal A_{\rho(x)}}W\,dx=0.
\end{align}
 uniformly w.r.t. $\varphi$ such that \eqref{h1boundary} and \eqref{h2boundary} are satisfied in $K_L(a)$.
\end{enumerate}
\end{cor}



\section{The study of 2D Coulomb Gases}\label{sec5}

In this section, we turn to the Coulomb gas minimizers, and we use the notation of the introduction. 
\subsection{Separation of points}
In this subsection we prove  item (1) of Theorem \ref{th3} and Theorem \ref{thlieb}. They rely on the same idea, quite independent from the rest of the paper: exploiting the fact that in the ground state of a Coulomb system, each point is at the minimum of the potential generated by the rest of the charges (i.e. the other points and the background charge).  These results be used in the rest of the proof of Theorem \ref{th3}.

The starting point is  the following:
\begin{lem}
	\label{lempropminpoints}
 Let $x_1, \dots, x_n$ be $n$ distinct points in $\mr^2$, $m_0'$ be as in Section \ref{sec1.3},
 $$H_n=- 2\pi \Delta^{-1} \( \sum_{i=1}^n \delta_{x_i'} - m_0'(x')\)= - \log * \( \sum_{i=1}^n \delta_{x_i'} - m_0'(x')\) , $$
 and $E_n= -\nab H_n$ as in  \eqref{En}, and let $U=H_n + \log |x-x_1'|$.  For any  point $y'\in \mr^2 $, letting 
  $$\tilde{E_n}=2\pi \nab \Delta^{-1}\( \sum_{i=2}^n \delta_{x_i'} + \delta_{y'} - m_0'(x')\),$$ it holds that
 \begin{equation}\label{51}
 W(E_n, \indic_{\mr^2})- W(\tilde{E_n}, \indic_{\mr^2})
 = U(x_1')- U(y').\end{equation}
 If $(x_1, \dots, x_n)$ minimizes $w_n$, then for every $y\in \mr^2$, $y'= \sqrt{n} y$, we have
 \begin{equation}
 \label{Uz}
 U(x_1') + 2n \zeta(x_1) \le U(y') + 2n \zeta(y).\end{equation}

 \end{lem}

\begin{proof}
We denote $\tilde{H_n}= - 2\pi  \Delta^{-1}\( \sum_{i=2}^n \delta_{x_i'} + \delta_{y'} - m_0'(x')\)$ so that $\tilde{E_n}= -\nab \tilde{H_n}.$ We note that $H_n$ and $\tilde{H_n}$ are well-defined by convolution with $-\log $ and, since they correspond to the potential generated by a zero total charge, they both decay like $1/|x|$ as $x \to \infty$ while their gradients decay like $1/|x|^2$.

We now let $\bar{H}(x) = \tilde{H_n}(x)- H_n(x)= - \log |x-y'|+ \log |x-x_1'|$. By definition of $W$, we have
\begin{multline}\label{decW1}
W(\tilde{E_n}, \1_{\mr^2})=
\lim_{\eta\to 0}\Big( \hal \int_{\mr^2 \backslash (\cup_{i=1}^n B(x_i', \eta) \cup B(y',\eta))} |\nab (H_n+ \bar{H})|^2 + \pi n \log \eta\Big)\\
= W(E_n, \1_{\mr^2}) 
+ \lim_{\eta\to 0}\Big( \hal \int_{\mr^2 \backslash (\cup_{i=1}^n B(x_i', \eta) \cup B(y',\eta))}|\nab \bar{H}|^2 \\+ \int_{\mr^2 \backslash (\cup_{i=1}^n B(x_i', \eta) \cup B(y',\eta))}\nab H_n \cdot \nab \bar{H}\Big),\end{multline}
where we have expanded the square.
We now turn to computing
\begin{multline}\label{52bis}
\lim_{\eta\to 0} \Big(\hal \int_{\mr^2 \backslash (\cup_{i=1}^n B(x_i', \eta) \cup B(y',\eta))}|\nab \bar{H}|^2 + \int_{\mr^2 \backslash (\cup_{i=1}^n B(x_i', \eta) \cup B(y',\eta))}\nab H_n \cdot \nab \bar{H}\Big)
\\
= \lim_{\eta\to 0}\Big( \hal \int_{\mr^2 \backslash (B(x_1', \eta) \cup B(y',\eta))}|\nab \bar{H}|^2 + \int_{\mr^2 \backslash (\cup_{i=1}^n B(x_i', \eta) \cup B(y',\eta))}\nab H_n \cdot \nab \bar{H}\Big).\end{multline} This relies on direct computations for renormalized energies, \`a la \cite{bbh}.
First, using Green's theorem and noting that $\bar{H}$ and $\nab \bar{H}$ also decay sufficiently fast at infinity, we have
$$\hal \int_{\mr^2 \backslash (B(x_1', \eta) \cup B(y',\eta))}|\nab \bar{H}|^2
= -\hal  \int_{\p B(x_1', \eta)} \bar{H} \nab \bar{H} \cdot \vnu
- \hal\int_{\p B(y', \eta)} \bar{H} \nab \bar{H}\cdot \vnu$$
where $\vnu$ is the outer unit normal to the balls.
Inserting $\bar{H}=   - \log |x-y'|+ \log |x-x_1'|$ and using again Green's theorem and the fact that $- \Delta \bar{H}=2\pi ( \delta_{y'}- \delta_{x_1'}) $, we find, if $y' \neq x_1'$, 
\begin{equation}\label{partiehb}
\hal \int_{\mr^2 \backslash (B(x_1', \eta) \cup B(y',\eta))}|\nab \bar{H}|^2
= - 2 \pi  \log \eta + ( 2\pi+o_\eta(1))  \log |x_1'-y'| .\end{equation}

On the other hand, using similar computations based on Green's theorem, we have, if $y'\notin\{x_1', \dots, x_n'\}$,
\begin{multline*}
 \int_{\mr^2 \backslash (\cup_{i=1}^n B(x_i', \eta) \cup B(y',\eta))}\nab H_n \cdot \nab \bar{H}
= - \sum_{i= 1}^n \int_{\p B(x_i', \eta)} \bar{H}\nab H_n \cdot  \vnu \\ -2 \pi  \int_{\mr^2} \bar{H}(x') m_0'(x') \, dx' +o_\eta(1),\end{multline*}
where we have used that $m_0'\in L^\infty$ and is compactly supported.
It follows with the explicit expression of $\bar{H}$ that
\begin{multline*}
 \int_{\mr^2 \backslash (\cup_{i=1}^n B(x_i', \eta) \cup B(y',\eta))}\nab H_n \cdot \nab \bar{H}\\
=  - 2\pi \sum_{i=1}^n \log |x_i'- y'|   + 2\pi \log \eta
 + 2\pi \sum_{i=2}^n \log |x_i'-x_1'| \\
 + 2\pi \int_{\mr^2} (\log |x-y'|- \log |x-x_1'|) \,m_0'(x') \, dx'+o_\eta(1).\end{multline*}
 Combining with the results of \eqref{decW1}, \eqref{52bis} and \eqref{partiehb},  we find that if $y'\notin \{x_1', \dots, x_n'\}$,
 \begin{multline*}
 W(\tilde{E_n}, \1_{\mr^2})- W(E_n, \1_{\mr^2})\\ =
  - 2\pi \sum_{i=2}^n \log |x_i'- y'|
 + 2\pi \sum_{i=2}^n \log |x_i'-x_1'| -  2\pi \int_{\mr^2} (\log |x-y'|- \log |x-x_1'|) \,m_0'(x') \,
 \\
 = U(y')- U(x_1').\end{multline*}

Of course, if $y'=x_1'$ then both sides equal zero and the result is true. If $y'\in\{x_2', \dots, x_n'\}$ we may also verify that $W(\tilde{E_n}, \1_{\mr^2})=+\infty$ and $U(y')=+\infty$, and the result also holds in a generalized sense.

Let us now turn to the application to $(x_1, \dots, x_n)$ minimizing $w_n$.
In that case, in view of the splitting formula \eqref{splitting},  by comparing with the energy of $(y, x_2, \dots, x_n)$, for any $y \in \mr^2$ and letting $y'= \sqrt{n} y$,  we have 
$$W(E_n,\1_{\mr^2})+ 2n \zeta(x_1) \le W(\tilde{E_n}, \1_{\mr^2}) + 2n \zeta(y).$$
Inserting \eqref{51},  it follows that  \eqref{Uz} holds.
\end{proof}

\begin{proof}[Proof of item (1) of Theorem \ref{th3}]
Let $(x_1 ,\dots, x_n)$ minimize $w_n$, and let $U$ be as in Lemma \ref{lempropminpoints}. 
Since $m_0'=0$ in $\mr^2 \backslash \Sigma'$,  by definition of $U$  we have that $-\Delta U\ge 0$ in $\mr^2 \backslash \Sigma'$, $U$ is thus superharmonic in that set. In addition $U \to + \infty$ as $x \to \infty$, since $H_n\to 0 $ as $ x \to \infty$. Thus, $U$ can only achieve its minimum on $\overline{\mr^2 \backslash \Sigma'}$ at some point $\bar{y}' \in \p \Sigma'.$ Since $\zeta=0 $ in $\Sigma $ and $\zeta \ge 0$ everywhere, we also have $\zeta(\bar{y}) \le \zeta(x_1)$, where $\bar{y} = \frac{\bar{y}'}{\sqrt{n}}$.
Comparing with \eqref{Uz} we obtain a contradiction, unless $x_1' \in \Sigma'.$
Since $w_n$ is symmetric in the labelling of the points, this proves that all the $x_i$'s must belong to $\Sigma$, concluding the proof.
\end{proof}

\begin{proof}[Proof of Theorem \ref{thlieb}]
We use the original argument of \cite{lieb}, adapted to the case of a varying background and a finite size $\E$. Let again $(x_1, \dots, x_n)$ minimize $w_n$.
First we note that now that we know that all the points $x_i'$ belong to $\Sigma$, \eqref{Uz} gives that for any $y\in \mr^2$,
\begin{equation}\label{Uz2}
U(x_1') \le U(y')+2n \zeta(y).\end{equation}

Let now $x_2$ (up to relabelling) be again a point in the minimizing collection $x_1, \dots, x_n$, and assume $x_1$ is its nearest neighbor in the collection (again, up to relabelling). We will use the minimality relation \eqref{Uz2} with respect to $x_1$.

We note that  since  \eqref{minom} holds, we may find $0<r_1<1$ such that
\begin{equation}\label{charge0}
\int_{B(x_2', r_1)} m_0'(x')\le \hal.\end{equation}
Let us split  $U$  into 
\begin{equation}
\label{splitu}
U= U^{near}  + U^{rem}+U^{corr}\end{equation}
where
\begin{align*}
& U^{near}= - \log * \Big( \delta_{x_2'}-  \1_{B(x_2', r_1)} \fint_{B(x_2', r_1)}  m_0'(x')\Big)
\\ &
U^{corr}= - \log *\(\1_{B(x_2', r_1)}\Big(-m_0'(x')+ \fint_{B(x_2', r_1)}  m_0'(x')\Big)\)\\
&
U^{rem}= -\log * \Big( \sum_{i=3}^n \delta_{x_i'}- m_0'(x') \1_{\mr^2 \backslash B(x_2', r_1)}\Big).\end{align*}

One observes that  $U^{near}$ is radial with respect to the origin at $x_2'$ that is $U^{near}(x')= f(|x'-x_2'|)$.
In fact $f$ can be computed explicitly, and with the choice \eqref{charge0}, it is decreasing, tends to $+\infty$ at $0$,  and $-\infty$  at $+\infty$. This implies that given any constant $M>0$ there exist $0< r_0<r_2<r_1$ such that \begin{equation}
\label{decrf}
\forall r<r_0, \quad f (r_2)<f(r)-M.\end{equation}

Let us now set $C_1=\sup_{\dist (x', \Sigma')\le 1} 2n \zeta(x)$. We claim that $C_1$ is bounded, independently of $n$.
In fact, $\zeta$ is related to a solution of an obstacle problem and it is stated in \cite[(1.17)]{sandierserfatygas} that $\zeta(x)\le
C\dist (x,  \Sigma)^2$ for some $C>0$, from \cite[Lemma 2]{caffarelli}.
It follows from this estimate that $\sup_{\dist (x', \Sigma')\le 1} 2n \zeta(x) \le 2 C$, hence the claim.

 Let us also observe that $U^{corr}$ is a bounded  function since $m_0'$ is,  and  set $C_2=2 \|U^{corr}\|_{L^\infty}$ (again  we could get a better estimate by using the fact that $|\nab m_0'|\le C/\sqrt{n}$).
We then set $M=C_1+C_2$ and have \eqref{decrf}. The constant $r_0$ depends only on the bounds in \eqref{minom} and the growth  of $\zeta$ away from $\Sigma$, so depends only on $V$.

If $|x_1'- x_2' |\ge r_0$ then we have  the desired conclusion. Assume thus that $|x_1'-x_2'|< r_0$.
 We  note that $U^{rem}$ is superharmonic in $B(x_2',r_2)$ since $r_2< r_1$. Let us then denote by $\bar{y}'$ the point in $\p B(x_2', r_2)$ where it achieves its minimum. Thus $U^{rem}(\bar{y}') \le U^{rem} (x_1').$ We also note that since $r_2<r_1<1$ and $x_2'\in \Sigma'$ by the result of item (1) of Theorem \ref{th3} proved above, we have $\dist (\bar{y}',\Sigma')  \le 1$ hence
 $2n\zeta(\bar{y})\le C_1$.
Since $|x_1'-x_2'| <r_0$, and in view of \eqref{decrf}, we have   $U^{near} (\bar{y}')<U^{near}(x_1')-M$. Combining the two, with  \eqref{splitu} it follows that
$$
U(\bar{y}')+ 2n\zeta(\bar{y}')< U^{near}(x_1')+ U^{rem} (x_1')+2n\zeta(\bar{y}')+ \hal C_2 -M  \le U(x_1')$$ in view of the choice of $M$ and the definitions of $C_1$ and $C_2$, a contradiction with \eqref{Uz2}.
This concludes the proof of Theorem \ref{thlieb}.\end{proof}

From now on we will use the fact that for a minimizer of $w_n$,  all the points are in $\Sigma$, and in view of \eqref{splitting} and the result of \cite[Theorem 2]{sandierserfatygas} we have the a priori bound
\begin{equation}
\label{bapriori}
W(E_n, \1_{\mr^2}) \le n \int_{\Sigma} \min_{\mathcal{A}_{\mu_0(x)}} W \, dx+o(n), \quad \text{as } \ n \to \infty.
\end{equation}

\subsection{Application of the previous analysis: end of the proof of Theorem \ref{th3}}\label{sec6}
We turn to the proof of item (2) of Theorem \ref{th3}, which  relies on applying the results that we obtained in Section \ref{sectionnonconstant}.
 We will use in particular that since   $m_0$ is $\mathcal C^1$ in $\E$, we have
\begin{equation}\label{eqgradm0}
\norm \nabla m_0'\norm_{L^\infty(\E')} \le \frac{C}{\sqrt n},
\end{equation}
and, whenever $K_L(a)\subset \E'$,  $L\le\sqrt{n}$, \begin{equation}\label{petitesvar}
\norm m'_0\norm_{\mathcal C^{0,\lambda}(K_L(a))}L^{\beta \lambda}\le C\frac{1}{(\sqrt{n})^{(1-\beta)\lambda}}\le o_n(1)\end{equation}
for all $\beta \in (0,1)$ and $0<\lambda\le 1$.
This will be inserted into the result of Corollary \ref{proplimminnc}.
We will also use that $\p \E$ is $\mathcal C^1$ so $\p \E'$ is locally almost flat, as $n\to \infty$.


Let $(x_1,\dots, x_n)$ minimize $w_n$ and $\j_n$ be as in \eqref{En}. As seen in Section \ref{sec5}, since all the points are in $\E$, in view of \eqref{splitting}, by minimality of $x_1, \dots, x_n$
we have 
\begin{equation}\label{compar}
W(E_n, \indic_{\mr^2}) \le  W(\tilde{E_n}, \indic_{\mr^2})
\end{equation}
for any $\tilde{E_n}= 2\pi \nabla \Delta^{-1} (\sum_{i=1}^n \delta_{y_i'} -  m_0' (x'))$ 
such that the $y_i'$ are all in $\E'$. 

We claim that  for any set $\Omega$,  $\j_n$ minimizes $W(E, \Omega)$ in the class $\mathcal{B}_{m_0' } (\Omega)$ with respect to its own boundary condition $\vp$, as long as competitors have all their points in $\E'$.
Indeed, let $\j$ be a competitor in $\Omega$ with $\j \cdot \vnu =\vp$ on $\p \Omega$ and all the points included in $\E'$. Then  consider $\bar{\j}$ the vector field equal to $\j $ in $\Omega $ and $\j_n$ in $\Omega^c$. 
We have 
$$W(\bar{\j}, \1_{\mr^2}) = W(\j, \1_{\Omega}) + W(\j_n, \1_{\Omega^c}).$$
Moreover, since no divergence is created at the boundary, $\bar{\j}$ still satisfies $\div \bar{\j} = 2\pi (\sum_i{\delta_{y_i'} } - m_0')$ with $y_i' \in \E'$.
Using Lemma \ref{lemdiv0}  (or rather its proof), we can modify $\bar{\j}$ to make it curl-free, while decreasing its energy. This gives a vector field $\tilde{\j}\in \mathcal B_{m_0'}(\mr^2)$  and with 
$$W(\tilde{\j}, \1_{\mr^2}) \le  W(\j, \1_{\Omega}) + W(\j_n, \1_{\Omega^c}).$$
But $W(\tilde{\j}, \1_{\mr^2}) \ge W(E_n, \indic_{\mr^2})$ by \eqref{compar}, so we must have $W(\j, \1_{\Omega}) \ge W(\j_n, \1_{\Omega})$, which proves the claim.

For the proof of  Theorem \ref{th3}, item (2),
  we note that  since  $K_\l(a)\subset \E'$ and    $\dist(K_\l(a),\p \E' )\ge   \sqrt{n}^{\beta}$,  the situation is essentially the same as in Theorem \ref{thmmaincb} with $L$ replaced by $\sqrt{n}$. 
We can apply the proof of Theorem \ref{thmmaincb} combined with the results  of Section \ref{sectionnonconstant}: we start from the initial scale $L=\sqrt{n}$ with the a priori bound \eqref{bapriori} and replace the initial square $K_L$ by $\E'$ itself, 
   and  we use the fact  that   $\j_n$ is  a minimizer on any square included in $\E'$ with respect to its own boundary conditions (by the claim above). Then 
it suffices to copy the proof of Theorem \ref{thmmaincb}, but in order to deal with the nonconstant background, replacing the use of Lemma \ref{propestimationgb} by Corollary \ref{proplimminnc}.  Note that the  constants  $C_1$ and $C_2$ can be chosen to be $C_1= \wb+1 $ and $C_2= \max_{m\in [\bw,\wb]} \min_{\mathcal{A}_m} W +1.$
This proves  \eqref{137}.
For \eqref{eqnumpointsgas}, the proof is identical to Step 2 of the proof of Theorem \ref{thmmaincb}, except that we use the fact that $\|m_0'\|_{L^\infty}$ is bounded by a constant.

This completes the proof of Theorem \ref{th3}.

\appendix

\section{Technical Results}
\begin{proof}[Proof of Lemma \ref{lemsysum}]
 The proof of this lemma is inspired from that of Lemma $4.15$ of \cite{sandierserfaty}.
We write the solution $u$ of \eqref{eqsysurho} as $u=u_1+u_2+u_3$ where
\begin{equation}\label{eqsysu1rho}
\left\{\begin{aligned}
-\Delta u_1 &=2\pi(m-\rho_{\mathcal R})& &\mbox{in}& &\mathcal{R}\\
\frac{\partial u}{\partial \vnu}&=\bar \varphi & & \mbox{on}& &\partial\mathcal{R}
\end{aligned}\right.
\end{equation}
where $\bar\varphi$ is equal to the average of $\varphi$ on the side where $\varphi$ is supported and is $0$ on the other sides;
\begin{equation}\label{eqsysu2rho}
\left\{\begin{aligned}
-\Delta u_2 &=0& &\mbox{in}& &\mathcal{R}\\
\frac{\partial u_2}{\partial \vnu}&=\varphi-\bar \varphi & & \mbox{on}& &\partial\mathcal{R}
\end{aligned}\right.;
\end{equation}
and \begin{equation}\label{eqsysu3rho}
\left\{\begin{aligned}
-\Delta u_3 &=2\pi(\rho_{\mathcal R}-\rho(x))& &\mbox{in}& &\mathcal{R}\\
\frac{\partial u_3}{\partial \vnu}&=0 & & \mbox{on}& &\partial\mathcal{R}
\end{aligned}\right..
\end{equation}
As in the proof of Lemma $4.15$ of \cite{sandierserfaty}, we have
\begin{equation}\label{eqnormu1rho}
\int_{\mathcal{R}}|\nabla u_1|^q\le C_{p,q}L^{2-\frac{q}{p}}\|\varphi\|^q_{L^p(\partial\mathcal{R})}
\end{equation}
and
\begin{equation}\label{eqnormu2rho}
\int_{\mathcal{R}}|\nabla u_2|^q\le C_{p,q}L^{2-\frac{q}{p}}\|\varphi\|^q_{L^p(\partial\mathcal{R})}
\end{equation}
where $C_{p,q}$ is a constant that depends only on $p$ and $q$.

Finally, by elliptic regularity and a scaling argument, we find that for any $q>1$
\begin{equation}\label{eqnormu3rho}
\norm\nabla u_3\norm_{L^q(\mathcal{R})}\le CL^{1+2/q}\norm \rho-\rho_\mathcal{R}\norm_{L^{\infty}(\mathcal R)}
\end{equation}
with $C$ is an universal constant (see \cite{grisvard} for more details on elliptic regularity).

Combining \eqref{eqnormu1rho}, \eqref{eqnormu2rho} and \eqref{eqnormu3rho}, we obtain \eqref{eqnormurho}.
\end{proof}


The following lemma serves to estimate the energy of a sum of vector fields.

\begin{lem}\label{lemsumvf}			
	Let $U\subset \RR^2$ be an open set and $\j_1$, $\j_2$ be two vector fields satisfying
	\begin{equation*}
		\left\{
		\begin{aligned}
			&\di E_1=2\pi (\sum_{p\in \Lambda}\delta_p- \rho_1)&&\mbox{in}\ U\\
			&\di E_2=\rho_2&&\mbox{in}\ U
		\end{aligned}
		\right.
	\end{equation*}
	for some discrete subset $\Lambda$ of $\RR^2$, and some nonnegative bounded functions $\rho_1$, $\rho_2$, and such that $\curl \j_1 $  vanishes in a neighborhood of each $p\in \Lambda$. If for some $q>2$ and $q'$ its conjuguate exponent, we have $E_1\in L^{q'}(U)$, $E_2\in L^{q}(U)$, then
	\begin{equation}
		\label{eqsumvf}
		W(E_1+E_2,\1_U)\le W(E_1,\1_U)+\frac{1}{2}\norm E_2\norm^2_{L^2(U)}+\norm E_1\norm_{L^{q'}(U)}\norm E_2\norm_{L^{q}(U)}
	\end{equation}
\end{lem}			

\begin{proof}
It suffices to write that 
\begin{multline*}\hal \int_{U\backslash \cup_{p\in\Lambda} B(p, \eta)} |E_1+E_2|^2\\=
\hal  \int_{U\backslash \cup_{p\in\Lambda} B(p, \eta)} |E_1|^2 +\hal  \int_{U\backslash \cup_{p\in\Lambda} B(p, \eta)}|E_2|^2 + \int_{U\backslash \cup_{p\in\Lambda} B(p, \eta)} E_1 \cdot E_2.\end{multline*}
We have $\hal  \int_{U\backslash \cup_{p\in\Lambda} B(p, \eta)}|E_2|^2 \le 
\hal 
\norm E_2\norm_{L^{2}(U)}^2$  and $$\int_{U\backslash \cup_{p\in\Lambda} B(p, \eta)} E_1 \cdot E_2\le \norm E_1\norm_{L^{q'}(U)}\norm E_2\norm_{L^{q}(U)}$$ by H\"older's inequality. The result then follows easily by adding $\# \Lambda \log \eta $ and letting $\eta \to 0$.
\end{proof}

To construct a vector field $\j$ which belongs to the admissible class, we need $\curl \j=0$. This can be done though the following lemma.

\begin{lem}\label{lemdiv0} Let $\Omega\subset\RR^2$ such that $\partial \Omega$ is Lipschitz. Let $p\in (1,2)$ and  $\varphi\in L^p(\partial \Omega)$. Let $\rho$ be a nonnegative $\mathcal C^0$ function. Let $\j$ be a vector field in $\RR^2$ such that
\begin{equation}
\left\{
\begin{aligned}
&\di \j=2\pi(\nu-\rho)&&\mbox{in } \Omega\\
&\j\cdot\vnu=\varphi&&\mbox{on }\partial \Omega\\
\end{aligned}
\right.
\end{equation}
where $\nu$ has the form
$$
\nu=\sum_{p\in\Lambda}\delta_p\quad\mbox{for some discrete set}\ \Lambda\subset \Omega
$$
and \eqref{eqcondnu} is satisfied. Then there exists $\tj\in\mathcal B_{\rho,\varphi}(\Omega)$ such that
\begin{equation}\label{eqendiv0}
W(\tj,\1_{\Omega})\le W(\j,\1_{\Omega})
\end{equation}
\end{lem}

\begin{proof}
We have to modify $\j$ so that $\curl\, \j=0$. Let $\zeta$ be the solution of
$$
\left\{
\begin{aligned}
&-\Delta \zeta= \mathrm{curl}\,\j&&\mbox{in } \Omega\\
&\zeta=0&&\mbox{on }\partial \Omega
\end{aligned}
\right.,
$$
and let $\tj=\j+\nabla^{\bot}\zeta$. We obtain
\begin{equation*}
\left\{
\begin{aligned}
&\di \tj=2\pi(\nu-\rho)&&\mbox{in } \Omega\\
&\curl\tj=0&&\mbox{in } \Omega\\
&\tj\cdot\vnu=\varphi&&\mbox{on }\partial \Omega\\
\end{aligned}
\right.,
\end{equation*}
since $\tj\cdot\vnu=\j\cdot\vnu$. Moreover,
\begin{align*}
 \int_{\Omega\backslash \bigcup B(p,\eta)}\vb \j\vb^2-&\, \int_{\Omega\backslash \bigcup B(p,\eta)}\vb \tj\vb^2=-2\int_{\Omega\backslash \bigcup B(p,\eta)}\tj\cdot\nabla^{\bot}\zeta\\
 &+\int_{K_l(a)\backslash \bigcup B(p,\eta)}\vb\nabla\zeta\vb^2\ge -2\int_{\Omega\backslash \bigcup B(p,\eta)}\tj\cdot\nabla^{\bot}\zeta.
\end{align*}
Since the term on the right-hand side converges as $\eta\to0$ to the integral over $\Omega$ and $\int_{\Omega}\tj\cdot\nabla^{\bot}\zeta=0$ (by applying Green's theorem and using the fact that $\curl\tj=0$ in $\Omega$ and $\zeta=0$ on $\partial \Omega$), we obtain
$$
W(\j,\1_{K_l(a)})-W(\tj,\1_{K_l(a)})\ge 0.
$$
\end{proof}

\vskip .2cm

\noindent
\sc Simona Rota Nodari\\
Laboratoire AGM, Universit\'e de Cergy-Pontoise\\
2 av. Adolphe Chauvin \\
95302 Cergy-Pontoise\\
{\tt simona.rota-nodari@u-cergy.fr}
\\

\noindent
\sc Sylvia Serfaty\\
UPMC Univ  Paris 06, UMR 7598 Laboratoire Jacques-Louis Lions,\\
 Paris, F-75005 France ;\\
 CNRS, UMR 7598 LJLL, Paris, F-75005 France \\
 \&  Courant Institute, New York University\\
251 Mercer st, NY NY 10012, USA\\
{\tt serfaty@ann.jussieu.fr}

\end{document}